\documentclass[11pt]{article}
\usepackage{amsmath,paralist,amsthm,comment,amssymb}
\usepackage{cite}
\usepackage{graphicx}
\usepackage{algorithm} 
\usepackage{braket}
\usepackage{newfloat,algcompatible,subfigure}
\DeclareMathOperator{\supp}{supp}
\DeclareMathOperator*{\argmin}{arg\,min}
\DeclareMathOperator*{\trian}{\triangle}
\usepackage{bbm}

\usepackage[T1]{fontenc}
\usepackage{kpfonts}

\newtheorem{theorem}{Theorem}

\newtheorem{lemma}[theorem]{Lemma}

\newtheorem{remark}[theorem]{Remark}

\newtheorem{definition}[theorem]{Definition}

\usepackage[left=0.75in,right=0.75in,top=1in,bottom=1in]{geometry}
\usepackage[size=small]{caption}
\usepackage{color}

\newcommand{\nix}[1]{}

\usepackage[hidelinks,colorlinks=true,urlcolor= blue,linkcolor = red,citecolor = red]{hyperref}

\begin{document}
\title{Decoding  toric codes on three dimensional simplical complexes}
\author{Arun~B.~Aloshious~and~Pradeep~Kiran~Sarvepalli  \\ Department of Electrical Engineering\\ Indian Institute of Technology Madras\\
        Chennai, India 600 036 }
        \date{}
\maketitle
\begin{abstract}
	
Three dimensional (3D) toric codes are a class of stabilizer codes with local checks and come under the umbrella of topological codes.
While decoding algorithms have been proposed for the 3D toric code on a cubic lattice, there have been very few studies on the decoding of 3D toric codes over arbitrary lattices. 
Color codes in 3D can be mapped to toric codes. 
However, the resulting toric codes are not defined on cubic lattice. They are arbitrary lattices with triangular faces. 
Decoding toric codes over an arbitrary lattice will help in studying the performance of color codes.
 Furthermore, gauge color codes can also be decoded via 3D toric codes. 
Motivated by this,
we propose an efficient algorithm to decode 3D toric codes on  arbitrary lattices (with and without boundaries). 
We simulated the performance of 3D toric code for cubic lattice under bit flip channel. We obtained a threshold of 12.2\% for the toric code on the cubic lattice with periodic boundary conditions.
\end{abstract}

\section{Introduction}

Three dimensional (3D) toric codes \cite{hamma05,castelnovo08} are a class of topological quantum codes.
They are defined on lattices  in three dimensions. 
Although the generalization of toric codes to 3D has been known for quite sometime, they have not been investigated as
much as their 2D counterparts. 
In fact, very little is known about 3D toric codes on lattices other than the cubic lattice. 
We do not know if the toric code on the cubic lattice is the best code among the 3D toric codes.
There are many aspects of these codes which merit further study \cite{bravyi11,vasmer18,alicki10}.
{\em In this paper, we focus on one of these aspects, namely,  the problem of efficiently decoding  3D toric 
	codes on arbitrary lattices.}
There are at least two compelling reasons to study this problem which we discuss below.

Most of the previous work on the decoding of 3D toric codes has been related to toric code on the cubic lattice. 
It was known for a long time that the minimum weight perfect matching algorithm could be used to efficiently decode the
phase flip errors on any 3D toric code \cite{dennis02}.
The main challenge in decoding a 3D toric code is the correction of the bit flip errors. 
For a special case like the cubic lattice, decoders have been proposed for bit flip errors \cite{duivenvoorden18,kulkarni18}.
For arbitrary lattices correcting bit flip errors was not addressed until recently \cite{kubica18,breuckmann18}.
So for the rest of the discussion we assume that we are discussing the decoder for bit flip errors. 
The problem poses interesting and  additional challenges not encountered in the 2D setting.
Efficient decoding algorithms for arbitrary lattices will also enable a comparative study of
various 3D toric codes to determine which code has the highest threshold.  

Another reason comes from the fact that   3D color codes can be mapped to 3D toric codes \cite{aloshious18,kubica15,kubica19}.
The toric codes that arise out of these mappings are not defined on cubic lattices. 
For instance, the mapping in \cite{aloshious18}   leads to toric codes on lattices where the faces are always triangular, in contrast to the cubic lattice where the faces are square. 
In order to decode the 3D color codes via these mappings, we must be able to decode 3D toric codes on arbitrary lattices. 
Further, not only the 3D color codes, but certain subsystem codes such as  the gauge color codes \cite{kubica15a,bombin15,brown16} can also be decoded via 3D toric codes \cite{aloshious18}. 
These codes provide yet another reason to study the decoding of 3D toric codes on arbitrary lattices.

{\em Contributions.}
In this paper, we propose a decoder for a 3D toric on an arbitrary lattice.
We consider lattices with periodic boundary conditions as well 
as those with boundaries. 
\nix{The results in \cite{sullivan90} can be utilized to develop decoders for bit flip errors. We use a different approach.}  
Our approach is inspired by the decoder for 2D toric codes proposed by Delfosse {\em et al.} \cite{delfosse17}.
It works by mapping the bit flip error model into erasure model and then decoding the error as an erasure error.
There are multiple challenges in making the generalization from 2D to 3D. 
Primarily,  the topological structure of error operators as well the stabilizer and logical operators in 3D is different from that of 2D.
Here we deal with surfaces and not chains (paths) as in the case of 2D. 
We will discuss the challenges and our contributions  in more detail in subsequent sections. 
Our decoder resulted in a threshold of $\gtrapprox 12.2$\% for the toric codes on the cubic lattice. 
We also simulated this decoder for a non-cubic lattice that  arises in the context of stacked color codes. 
Here we observed a threshold of about $\approx13.3\%$ for bit flip channel.

{\em Previous work.}
We now briefly discuss  related  previous work for arbitrary lattices. 
Kubica {\em et al.} \cite{kubica18}  proposed a generalization of the Toom's rule for decoding toric codes on $D$-dimensional lattices. 
This has been called the sweep decoder therein.
The sweep decoder is applicable to toric codes on $D$-dimensional lattices where $D\geq 2$, while our decoder is specifically for 3D toric codes. 
However, the sweep decoder  can be applied only for certain lattices called causal lattices. 
It is possible to check whether a lattice is causal when it is translation invariant. 
it has been noted in \cite[Appendix~B]{kubica18} that it is ``challenging'' to verify if an arbitrary lattice is causal. 
Our decoder is applicable for noncausal lattices also as our algorithm does not make the assumptions made for a lattice to be causal. 

Breuckmann {\em et al.}  \cite{breuckmann18} proposed a neural network based decoder for the 3D toric code.
In principle, it is possible to use neural network based decoders proposd by them  for any stabilizer code, but a numerical study maybe required to validate its usefulness and performance.
For codes of short length, neural decoders offer an attractive alternative. 
The decoding complexity of the neural decoder is cubic while the proposed decoder is quadratic, but it can be made almost linear. 
Furthermore, for longer codes and lattices that are not translation invariant, neural network decoders could have significant training cost affecting the complexity of the decoder.

Sullivan~\cite{sullivan90} proposed a linear programming approach for finding an oriented minimal surface
given an oriented boundary. 
As noted by the authors of ~\cite{duivenvoorden18} these results could be used to decode 3D toric codes. 

On the cubic lattice we obtain a threshold of {$\gtrapprox 12.2\%$,} while Kubica {\em et al.} obtained a threshold of 14.2\%.
A higher threshold of 17.2\% was obtained by Duivenvoorden {\em et al.} \cite{duivenvoorden18} while
Breuckmann {\em et al.} obtained 17.5\%   possibly, because these decoders are tailored to the cubic lattice.

{\em Overview.}
This paper is organized as follows: In Section~\ref{sec:bg}, we briefly review 3D toric codes. 
In Section~\ref{sec:intuition}, we describe the intuition behind our decoder. 
We develop a decoding algorithm for the 3D toric codes with boundaries in Section~\ref{sec:3d-bndry}. 
In Section V, we generalize the decoder to toric codes with periodic boundary conditions.
Finally, in Section VI, we discuss the simulation details and complexity of our algorithm.
\section{Background}\label{sec:bg}
The single qubit Pauli group $\mathcal{P}$ is defined as follows:
\begin{eqnarray}
	\mathcal{P} = \{i^cX^\alpha Z^\beta\mid  c \in \{0,1,2,3\} ; \alpha,\beta  \in \{0,1\}\},  \label{eq:P1}
\end{eqnarray}
where  $X$, $Y$  $Z$ are the  Pauli matrices given by
\begin{eqnarray}
	X= \left(\begin{array}{cc} 0&1 \\1&0 \end{array} \right), Y= \left(\begin{array}{cc} 0&-i \\i&0 \end{array} \right),  \mbox{ and } Z= \left(\begin{array}{cc} 1&0 \\0&-1 \end{array} \right)\label{eq:XZ}
\end{eqnarray}

A Pauli operator on $n $ qubits is given by $E=E_1 \otimes E_2 \cdots \otimes E_n$ with $E_j \in \mathcal{P}$ for all $j$.
Let $\mathcal{P}^n$ denote the collection all such $n$-qubit Pauli operators.
Pauli operators either commute or anti-commute,  hence, for any $P, P' \in \mathcal{P}^n$, $PP'= \pm P'P $.
We define the support of  a Pauli  operator $E$ to be the set of qubits on which the operator acts nontrivially.  
\begin{eqnarray}
	\supp(E) = \{i ; E_i \neq I \}\label{eq:err-supp}
\end{eqnarray}

The weight of an operator $E$ is the number of qubits over which $E$ acts nontrivially.
\begin{eqnarray}
	\text{wt}(E)  = |\supp(E) |\label{eq:err-wt}
\end{eqnarray}
A single bit (phase) flip error on $i$th qubit is denoted as $X_i$ ($ Z_i$). 

A quantum code code $\mathcal{Q}$ is a subspace of the $2^n$-dimensional complex vector space $\mathbb{C}^{2^n}$.
Quantum stabilizer codes are a class of codes defined by an abelian subgroup $\mathcal{S} \subset \mathcal{P}^n$,
where  $-I \notin \mathcal{S}$. 
In this case the codespace $\mathcal{Q} $ is defined as follows:
\begin{equation}
	\mathcal{Q} = \left\{\ket{\psi} \in \mathbb{C}^{2^n}\mid S \ket{\psi} = \ket{\psi} \mbox{ for all } S \in \mathcal{S}\right\}. \label{eq:stab-code}
\end{equation} 
The subgroup $\mathcal{S}$ is called the stabilizer of the code and elements of $\mathcal{S}$
stabilizers. 
Each stabilizer is like a parity check. 
Any Pauli error $E$ commutes or anti-commutes with the stabilizer elements. If the error $E$ commutes with the stabilizer, then we say the syndrome corresponding to that stabilizer is zero. If $E$ anti-commutes with a stabilizer then the syndrome corresponding to that stabilizer is one. 

The elements from the centralizer of $\mathcal{S}$, i.e., the subgroup of Pauli operators which commute with all the elements of $\mathcal{S}$, cannot be detected by any of the stabilizers. These operators are either stabilizers which act trivially or logical operators which act 
nontrivially in the code space. 
When performing error correction, it is not necessary to find the exact error that occurred. 
It is sufficient to find an error that differs by a stabilizer. 
The exact decoding problem is to find the most probable class of errors that is consistent with the observed syndrome. 
In practice, we often try to find the most likely error up to a stabilizer generator. 
We refer the reader to \cite{gottesman97,calderbank98} for more details on stabilizer codes. 

\subsection{3D Toric code}
Let $\Gamma$ be a lattice in 3D. 
We denote the  vertices of $\Gamma$  as $C_0(\Gamma)$, edges as $C_1(\Gamma)$, faces as $C_2(\Gamma)$ and volumes as $C_3(\Gamma)$. 
The set of edges in the boundary of any face $f$ is given by $\partial(f)$.
We also denote the set of faces in the boundary of a volume $\nu$ by $\partial(\nu)$, and the set of vertices in boundary of edge $e$ is given by $\partial(v)$. 
Let $\iota(e)$ denote the faces incident on $e$.
We denote by  $A \triangle B$  the symmetric difference between two sets $A$ and $B$ given by 
\begin{eqnarray}
	A \triangle B = \{a \in A\cup B; a \notin A \cap B\}. \label{eq:aDb}
\end{eqnarray}
Note that the symmetric difference is associative. 
We can extend the operators $\partial$ and $\iota$ to appropriate collections of edges, faces and volumes. 
For a set of volumes $V$, set of  faces $F$ and set of edges $E$, we define the boundary operators $\partial(V)$,  $\partial(F)$ and co-boundary $\iota(E)$ as 
\begin{eqnarray}
	\partial(V) &=& \trian\limits_{\nu \in V} \partial(\nu), \\
	\partial(F) &=& \trian\limits_{f \in F} \partial(f), \\
	\iota(E) &=& \trian_{e \in E} \iota(e).
\end{eqnarray}

A 3D toric code on is a stabilizer code defined on a 3D lattice $\Gamma$.
Qubits are placed on faces of $\Gamma$.  
Two types of stabilizer generators are defined, one of $X$ type attached to volumes and one of $Z$ type attached to 
faces.
For each edge $e$ in $C_1(\Gamma)$, we have a $Z$ type stabilizer $B_e$ and for each volume $\nu $ in $C_3(\Gamma)$, we have a $X$ type stabilizer $A_{\nu}$ where $B_e$ and $A_{\nu}$ are given by
\begin{equation}
	B_e = \prod\limits_{f: e \in \partial(f)} Z_f = \prod\limits_{f \in \iota(e)} Z_f \label{eq:z-stab}
\end{equation}
and 
\begin{equation}
	A_{\nu} =  \prod\limits_{ f \in \partial(\nu)} X_{f},\label{eq:x-stab}
\end{equation}
where $X_f$ and $Z_f$ are  Pauli operators acting on the  qubit placed on the face $f$ and identity elsewhere. 
The stabilizer generators of the toric code on the cubic lattice 
are shown in Fig.~\ref{fig:stab}.
\begin{figure}[H]
	\centering
	\subfigure[]
	{\includegraphics[scale=1]{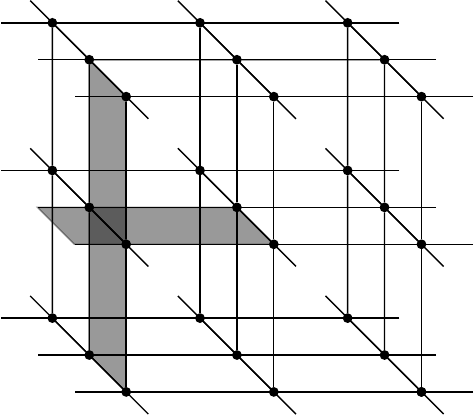}}
	\subfigure[]
	{\includegraphics[scale=1]{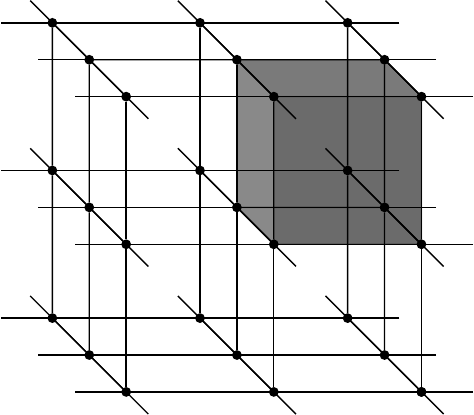}}
	\caption{Stabilizer generators of a 3D toric code. We allow for boundaries. Qubits are  placed on faces. 
		The faces can have half edges in their boundary. (a) A $X$ type stabilizer attached to an edge. The shaded faces show the support of the stabilizer. 
		These are faces incident on that edge. (b) A $Z$ type stabilizer attached to a volume. The support of the stabilizer is the boundary of the volume.}
	\label{fig:stab}
	
\end{figure}

The decoding problem on a toric code can be split into decoding of bit flip ($X$) errors and phase flip ($Z$) errors separately, as it is a Calderbank-Steane-Shor (CSS) code \cite{calderbank98}. 

Decoding of $Z$ errors in 3D toric code can reduced to the graph matching problem in a complete  graph. 
Efficient algorithms are known which perform well.
Therefore, we focus on the decoding of the bit flip errors in this paper. 

Bit flip errors are detected by $Z$ type stabilizer generators  $B_e$.
A single qubit error on a face causes nonzero syndrome on the boundary of the face. 
So the syndrome will be nonzero on the edges in $\partial (f)$.
Since qubits can be uniquely associated to faces, we also say the boundary of a qubit without confusion. 

In general,  a bit flip error $E$ creates a nonzero syndrome on the edges of the lattice. 
These are precisely the edges $e$ whose stabilizer generators $B_e$ anticommute with the error $E$.
A collection of edges $\sigma$ form a cycle if the vertices  incident on the edges have an even number of edges from the collection. 
A cycle is said to be simple if all the vertices in the cycle have degree two and are connected.

Under periodic boundary conditions, the nonzero syndrome corresponds to a collection of edges which form a union of (simple) cycles in the lattice.
In case there are boundaries, the syndrome can also be nonzero on edges which do not form a cycle.

Since $Z$ type stabilizers and $Z$ type logical operators commute with all the $X$ type stabilizers, they will produce a nonzero syndrome. 
So an error $E$ and $EM$, where $M$ is a stabilizer or a logical operator will lead to the same syndrome.

The decoding problem is to estimate the error which caused the observed syndrome. 
In graphical terms this is equivalent to finding a set of faces whose boundary is the collection of edges with nonzero syndrome.  
The estimate need not be the exact set of faces on the error occurred. 
It suffices if the estimate is  equivalent to the original error up to a stabilizer. 
In graphical terms this means that  the boundary of error estimate and the original error is a closed volume of trivial homology. The minimum weight decoder for bit flip error in 3D toric code is described as follows:

\begin{equation}
	\hat{\mathcal{E}} = \argmin\limits_{F \subseteq C_2(\Gamma)} |F| \text{ such that } S_E = \partial(F) \label{eq:dec-problem}
\end{equation}

\section{Intuition behind the decoding algorithm}\label{sec:intuition}
Before we give the complete decoder, we give the intuition behind the decoding algorithm and illustrate the main ideas 
by considering a simple error pattern. 
In a 3D toric code with periodic boundary condition, the edges carrying a nonzero syndrome form a collection of simple cycles of trivial homology. 
They are precisely the edges which form the boundary of faces (qubits) which are in error. 
Specifically, we consider the case when the syndrome is a 
simple cycle. 
Suppose now that $\sigma$ is a simple cycle and the boundary of an error  $E$. 
(The boundary of an error is the boundary of the support of that error i.e.  $\partial(E)=\partial(\supp(E))$.)

The decoding problem is to estimate an error that is equivalent to $E$ up to a stabilizer. 
In graphical terms this means we have to find a surface whose boundary is $\sigma$ and which is homologically equivalent to the original error. 
A simpler problem is to estimate an error that has the lowest 
weight (support) and having the same boundary. 

In general, there are an exponentially many errors with the same syndrome. More precisely, we have $2^{n_x+k}$ $X$ type errors have same syndrome, where $n_x$ is number of independent $X$ type stabilizer generators and $k$ is the number of encoded qubits. 
Here the number of stabilizer generators $n_x$ grows 
linearly with the length of the code.
Given the boundary of the error i.e., the nonzero syndrome, there are four components to our decoding algorithm.
The corresponding steps are repeated to obtain an error estimate. 
\begin{compactenum}[i)]
	\item Exploring potential qubits in error.
	\item Freezing error on certain qubits.
	\item Terminate exploration.
	\item Peeling: Iterative estimation of error on qubits. 
\end{compactenum}

In the first step of our algorithm we identify the set of qubits, denoted $\mathcal{E}$,  that could have caused the measured syndrome. 
At the least, this set must contain the qubits which participate in failed checks. 
So, we initialize $\mathcal{E}$ with the set of faces (qubits) incident on the edges with the nonzero syndrome. 
However, this set of qubits need not contain $E$ or an equivalent error consistent the syndrome. 
So we need to explore further, until we find a set of qubits which completely explains the syndrome. 
This motivates us to explore the qubits further, and identify qubits which could explain the syndrome and then estimate the error on them.

Suppose we have a set of potential candidates $\mathcal{E}$. 
If this set supports a stabilizer, then there are two or more solutions which are equivalent homologically. 
What this means is that some errors on some qubits in  $\mathcal{E}$ can be set to zero. 
However, we do not choose these qubits arbitrarily. 
Suppose that $\varsigma \subseteq \mathcal{E}$ is the set of qubits which support a stabilizer and $q\in \varsigma$
be the last qubit in that was added to $\mathcal{E}$.
We set the error on  $q$ to be zero. 
We call this process  freezing. 

What freezing also does is to reduce the number of potential solutions.
Every qubit that is frozen will reduce the number of potential solutions by a factor of two. 
Instead of considering only stabilizers in the support of $\mathcal{E}$, we also consider logical operators.
The reason for this is that if there is a logical operator in the support of $\mathcal{E}$, once again the error can be explained in two different ways. 
We repeat this process until all the qubits have been considered as potential candidates for $\mathcal{E}$.

Suppose that there are no more qubits left to explore.
At this stage the error is unique, because all solutions equivalent up to a stabilizer or a logical operator have been removed.
There is exactly one assignment of errors that will explain the syndrome. 
Mathematically speaking, at this juncture, we have a system of linear equations with a unique solution. 
Such a system of equations can be solved by first identifying the equations which involve exactly one variable. 
Then back substituting this variable's solution in other equations.
This corresponds to the last step of the decoding, namely, peeling. 
Thus we  estimate the error iteratively.

\begin{figure}
	\centering
	\subfigure[]
	{\includegraphics[scale=0.5]{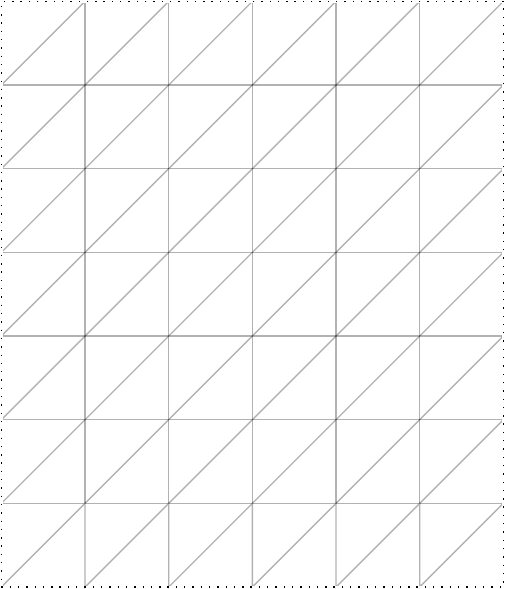}}
	\subfigure[]
	{\includegraphics[scale=0.5]{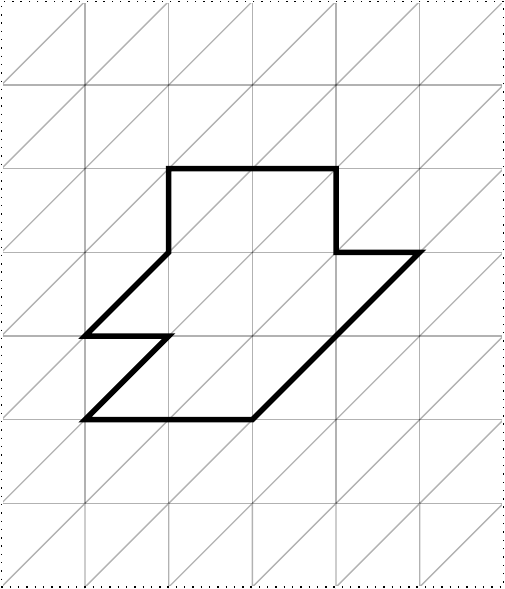}}
	\subfigure[]
	{\includegraphics[scale=0.5]{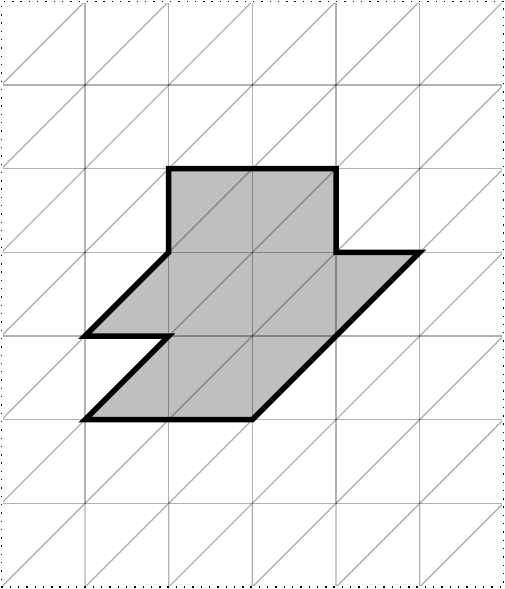}}
	\subfigure[]
	{\includegraphics[scale=0.5]{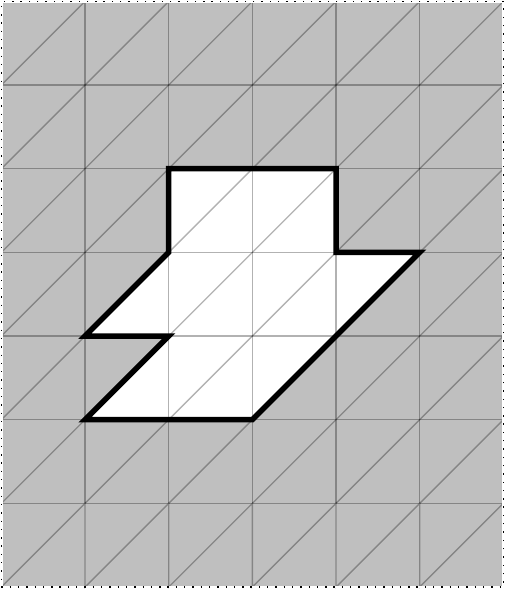}}
	\subfigure[]
	{\includegraphics[scale=0.5]{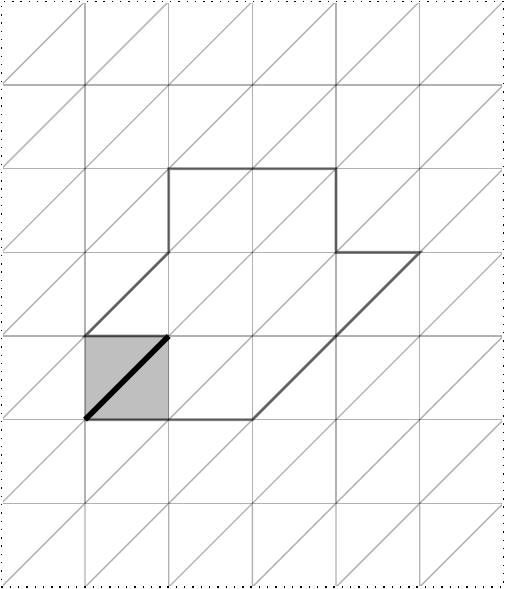}}
	\subfigure[]
	{\includegraphics[scale=0.5]{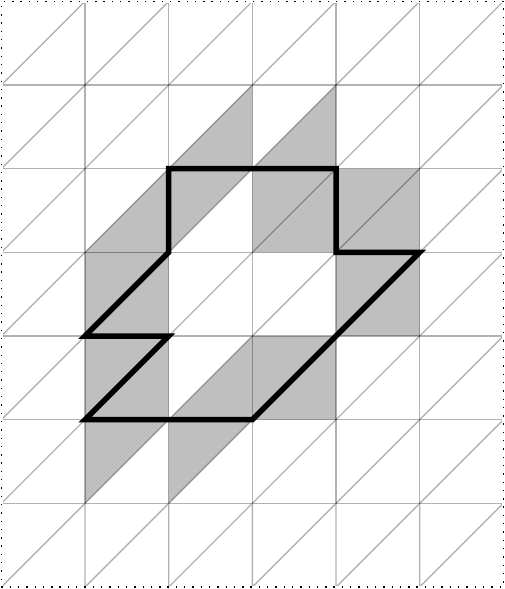}}
	\subfigure[]
	{\includegraphics[scale=0.5]{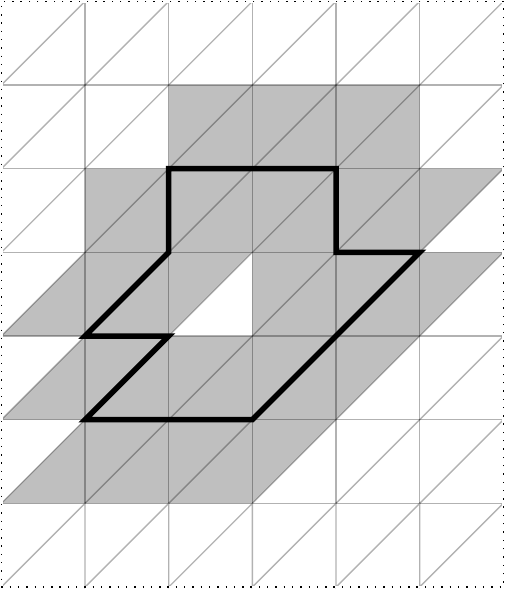}}
	\subfigure[]
	{\includegraphics[scale=0.5]{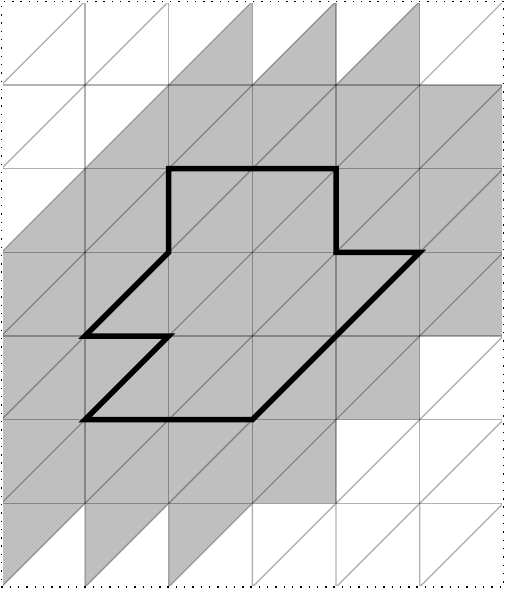}}
	\subfigure[]
	{\includegraphics[scale=0.5]{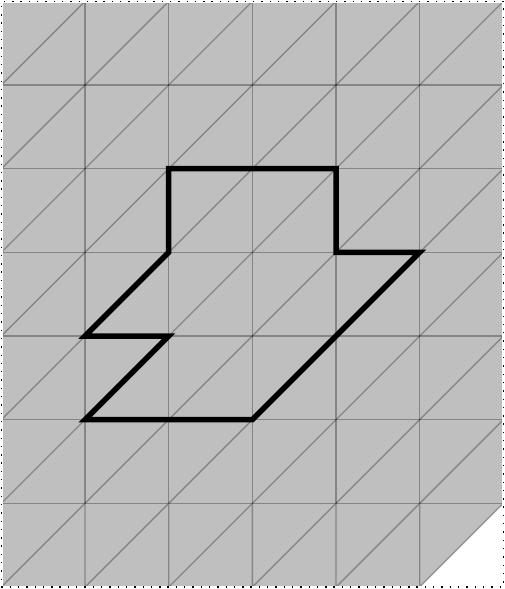}}
	\caption{(a) Two dimensional perspective of the decoder in which qubits are on faces and checks on edges. (b) a valid syndrome (Thick edges shows non zero syndrome and other edges with zero syndrome), (c) and (d) are two possible error support that are consistent with the non zero syndromes shown in (b). (e) shows faces that may be in error because of highlighted non zero syndrome. (f),(g) and (h) are consecutive addition of faces to potential solutions. In (h), we have one solution within the support but we do not have a simple test to verify. So continue growing and reach till (i) after which the last face is marked as not in error because adding that face will make two solution in the support.}
	\label{fig:intuition}
\end{figure}

 \section{Decoding 3D toric codes with boundaries}\label{sec:3d-bndry}
 
 In this section we focus on 3D toric codes with boundaries, see Fig.~\ref{fig:stab} for an example of such a code.
 The decoding ideas can be presented somewhat more simply for this case.
 First, we will take a closer look at the 3D toric code with boundaries. 
 Then following the discussion in previous section, we  identify a set of faces $\mathcal{E}$ such that there is exactly one error supported in $\mathcal{E}$ whose syndrome matches with the measured syndrome. 
 To find this set we start with an empty set and add faces to $\mathcal{E}$ such that  the set contains exactly one solution until all the qubits have been considered. 
 \subsection{Toric codes with boundaries}

 In this section we consider 3D toric codes with boundaries
 from a slightly different perspective than considered previously in literature.
 This perspective is useful in developing the decoding algorithm.
 These codes do not have periodic boundary conditions in any direction. 
 
 The ideas presented here for these codes could be applied for codes with a combination of boundaries and periodic boundaries.

 \noindent
 {\em Assumptions on the lattices. }
 One can encode information by punching holes in 
 the lattices. 
 This creates boundaries in the interior of the lattice.

 In this paper  we assume that the following conditions on the lattice $\Gamma$.
 \begin{compactenum}[($\mathsf{L}$1)]
 	\item $\Gamma$ does not have any boundaries in the interior.
 	\item The boundary of any  face in $\Gamma$and  $\Gamma^*$ 
 	is  a closed path or an open path starting and ending with partial edges. 
 	See Fig.~\ref{fig:faceingamma} for an illustration. 
 	
 \end{compactenum}
 Note that we allow for partial edges in $\Gamma$. 
 On the other hand, because of the assumption (L2) we do not allow for 
 faces of the form shown in 
 Fig.~\ref{fig:nofaceingamma} are not allowed. 
 Specifically, we do not allow when the boundary consists i)  of a disjoint set of paths, ii)  of an open path that terminates on vertices. 
 
 \begin{figure}[H]
 	\centering
 	\includegraphics[scale=1.5]{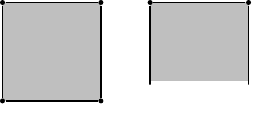}
 	\caption{ Allowable faces  in $\Gamma$ and $\Gamma^*$.
 		Note that a face can have partial edges in its boundary. 
 		If a face has partial edges, then they must exactly two. A partial edge is incident on exactly one vertex.}
 	\label{fig:faceingamma}
 \end{figure}
 \begin{figure}[H]
 	\centering
 	\includegraphics[scale=1.5]{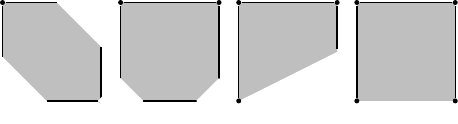}
 	\caption{Disallowed faces in $\Gamma$ and $\Gamma^*$.}
 	\label{fig:nofaceingamma}
 \end{figure}

 We say two cells are adjacent to each other if they share a face. 
 We say two volumes $\nu_0$  and $\nu_m$ are connected if there is a sequence of volumes $\nu_1$, \ldots, 
 $\nu_i, \nu_{i+1}$, \ldots, $\nu_{m-1}$ such that $\nu_i$ and $\nu_{i+1}$ are adjacent to each other for $0\leq i\leq m-1$.
 Let $f_i$ be a face adjacent to $\nu_{i-1}$ and $\nu_i$.
 The sequence of faces $\{ f_i \}_{i=1}^m$ is called a face-path connecting $\nu_0$ and $\nu_m$.

 \begin{definition}[Face path]
 	A face path $\rho$ is a sequence of faces  $\{ f_i\}_{i=1}^m$
 	along with sequence of volumes $\Lambda(\rho) =  \{\nu_i\}_{i=1}^{m-1} $ 
 	such that  $f_i,f_{i+1} \in \partial(\nu_i)$ and for all $i \neq j, \nu_i \neq \nu_j $. 
 \end{definition}
 
 We say $\rho$ is a face path from face $f_1$ to face $f_m$. We also say $\rho$ is a face path  from volume $\nu_0$ to $\nu_m$ if $f_1 \in \partial(\nu_0)$, $f_m \in \partial(\nu_m)$ and $\nu_0,\nu_m \neq \nu_i$ for $1 \leq i < m$. If $\nu_0=\nu_m$, then we say the face path is a simple face cycle. 
 A more general notion of face cycles is as follows. 
 
 \begin{definition}[Face cycle]
 	We call a sequence of faces $\rho$ as a face cycle if for all volumes $\nu$, the faces of $\partial(\nu)$ occurs even number of times in $\rho$ and for all faces $f$ in $\rho$ , $|\iota(f)|=2$.
 \end{definition}
 While a simple face cycle is always a face path, in general it is not necessary that a face cycle is a face path. 
 A face cycle is a face path if and only if such a face cycle is a simple face cycle.
 
 The assumptions on the lattice which we stated earlier can be also formulated in terms of the face paths.
 This reformulation makes some of the proofs easier. 
 The assumptions on (faces of) $\Gamma^\ast$ imply that 
 for all edges $e$ in $\Gamma$,
 the set of faces $\iota(e)$ can be rearranged to be either a face cycle as in Fig.~\ref{fig:edgeopisfacepath} or a face path as in Fig.~\ref{fig:edgeopisfacepath}.
 We will not have edges of the form shown in 
 Fig.~\ref{fig:edgeopisnotfacepath}.
 \begin{figure}[H]
 	\centering
 	\includegraphics[scale=0.8]{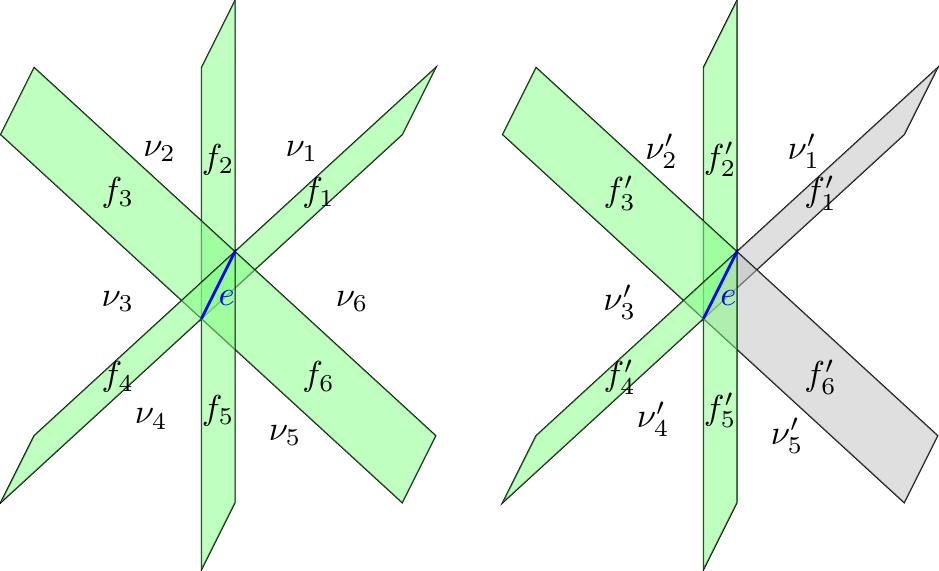}
 	\caption{
 		The assumption ($\mathsf{L}2$) on the lattice  implies that the edges of $\Gamma$ must take the form shown in the figure. More precisely, 
 		$\iota(e)$ must be a face path up to rearrangement of the faces. 
 		The green faces are incident on two volumes while the gray faces are incident only on one volume. 
 		In $\iota(e)$, there are either zero or exactly two gray faces. 
 	}
 	\label{fig:edgeopisfacepath}
 \end{figure}
 \begin{figure}[H]
 	\centering
 	\includegraphics[scale=0.8]{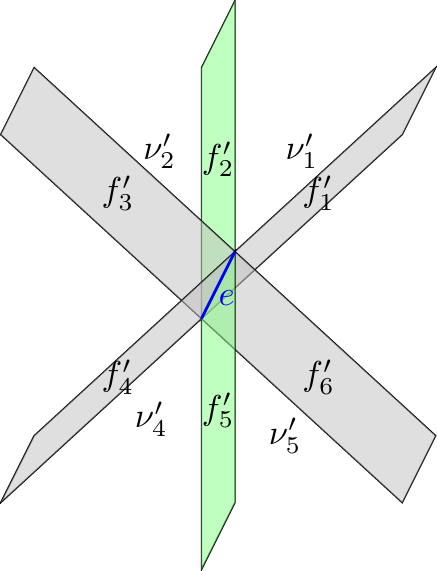}
 	\caption{The assumption ($\mathsf{L}2$) disallows certain edges, such as the edge $e$. Green faces are incident on two volumes. Gray faces are incident on exactly one volume. Edges like $e$ are not allowed in $\Gamma$ because in $\Gamma^*$, the face corresponding to this edge will violate ($\mathsf{L}2$).}
 	\label{fig:edgeopisnotfacepath}
 \end{figure}
 
 If $\iota(e)$  is a face path, then it is a face path between two faces both of which are incident on exactly one volume.
 Therefore, $\iota(e)$ contains 
 exactly zero or two faces which are incident on only one volume. 
 
 \begin{remark}\label{rm:edgefacepath}
 	If the lattice $\Gamma$ satisfies the condition ($\mathsf{L}2$), then for any edge $e \in C_1(\Gamma)$,  the set of faces in $\iota(e)$ can be arranged to form a face path. So we say that $\iota(e)$ is a face path.
 \end{remark}
 
 No boundaries in the interior of $\Gamma$ implies that  any face cycle $\rho$ in $\Gamma$  is the set of faces incident on a collection of edges   $ \rho = \iota(A)$,  for some   $A \subseteq C_1(\Gamma)$. 
 
 \begin{remark}
 If $\Gamma$ satisfies ($\mathsf{L}1$), then any face cycle $\rho$ in $\Gamma$  implies $ \rho = \iota(A)$,  for some   $A \subseteq C_1(\Gamma)$. 
 \end{remark}
 
 Since a face is incident on at most two volumes, we can make the following  observation. 
 \begin{remark}\label{rem:fp-f}
 	Any face in the face path $\rho$ is not incident on any volume other than the volumes in $\Lambda(\rho)\cup \{ \nu_0, \nu_m\}$.
 \end{remark}
 Note that there exist faces which are incident only on one volume. In a face path, these faces can occur only as $f_1$ or as $f_m$. If $f_1$ is incident on only one volume, then $\nu_0$ does not exist.
 This is because $f_1$ is already incident on $\nu_1$ by definition,
 therefore it cannot be incident on $\nu_0$.
 Similarly, if $f_m$ is incident on only one volume, then $\nu_m$ does not exist.
 \begin{remark}\label{rem:fp-vol}
 	An even number of faces from the face path are incident on any volume of $\Gamma$ except $\nu_0$ and $\nu_m$ provided $\nu_0$ and/or $\nu_m$ exist.
 \end{remark}
 \begin{remark}\label{rem:fp-in-dual}
 	If two faces are connected in $\Gamma$, then 
 	the edges associated to those two faces form a path in $\Gamma^\ast$ with $f$ and $f'$ being the first and last edges of that path. 
 	As a consequence, if
 	$\Gamma^\ast$ is connected, then every pair of faces in $\Gamma$ are also connected by a face path. 
 \end{remark}

 Faces which are incident on only one volume play an important role in determining the number of logical qubits encoded by the code.
 They can also used to define the $X$ and $Z$  logical operators. 
 For this reason, we will look at them closely.
 Let $\mathcal{F}$ be the set of faces, where each face is incident on only one volume of $\Gamma$.
 \begin{equation}
 	\mathcal{F} = \{f \in C_2(\Gamma) :  |\iota(f)| =1\}
 	\label{eq:mathcalf}
 \end{equation}
 where $\iota(f)$ is the set of volumes $\nu \in C_3(\Gamma)$ such that $f \in \partial(\nu)$.

 We now define an equivalence relation on $\mathcal{F}$. 
 This allows us to define a canonical set of logical operators.
 Suppose  a sequence of faces $\rho = (f=f_1, f_2, \hdots, f_m=f')$ forms a face path between two faces $f,f'$ 
 We denote the operator 
 \begin{eqnarray}
 	W_{f,f'}^Z (\rho)= \prod\limits_{i=1}^m Z_{f_i}\label{eq:face-path-Z-op}
 \end{eqnarray}
 Similarly, we can define the following operator for a face  path $\rho=(f_1, \ldots, f_m)$ connecting two volumes $\nu$ and $\nu'$.
 \begin{eqnarray}
 	W_{\nu,\nu'}^Z  (\rho)= \prod\limits_{i=1}^m Z_{f_i}\label{eq:face-path-Z-op-vol}
 \end{eqnarray}
 
 These  operators $W_{f,f'}^Z$ will be either logical operators or stabilizers if $f$ and $f'$ are in $\mathcal{F}$. 
 Lemma~\ref{lm:logop-stab-face-path} makes this more precise.
 \begin{lemma}[Operators from face paths between boundaries]
 	\label{lm:logop-stab-face-path}
 	Let a 3D toric code be defined on a lattice $\Gamma$, and  $\mathcal{F}$ be defined as in Eq.~\eqref{eq:mathcalf}.
 	For any pair of faces  $f,f ' \in \mathcal{F}$, and 
 	a  face path $\rho$ from $f$ to $f'$, the operator $W_{f,f'}^Z(\rho)$,  as in Eq.~\eqref{eq:face-path-Z-op},
 	is either a stabilizer or a logical operator.
 \end{lemma}
 \begin{proof}
 	From Remark \ref{rem:fp-vol}, all volumes except $\nu_0$ and $\nu_m$ will have even number of faces from 
 	$\rho$. Since $f_1,f_m \in \mathcal{F}$, $\nu_0$ and $\nu_m$ do not exist for $\rho$. Thus any volume will have even number faces from $\rho$. Hence all volume stabilizers will commute with the operator $W_{f,f'}^Z(\rho)$. 
 	Therefore, $W_{f,f'}^Z(\rho)$ must be a stabilizer or logical operator. 
 \end{proof}

 \begin{lemma}[Operators from simple face cycles]\label{lm:face-cycle-op}
 	Let $\rho$ be a simple face cycle from $f$ to $f'$ where 
 	$f,f'$ are in boundary of same volume. 
 	Then $W_{f,f'}^Z$ is a stabilizer or a logical operator. 
 \end{lemma}
 \begin{proof}
 	Any volume $\nu_i \in \Lambda(\rho)$ contains even number of faces $f_i,f_{i-1} \in \partial(\nu_i)$,  hence all $A_{\nu_i}$ commute with  $W_{f,f'}^Z(\rho)$. 
 	All other volumes of $\Gamma$ do not have any faces of $\rho$. Hence all volume stabilizers $A_{\nu}$ commute with the operator $W_{f,f'}^Z(\rho)$. 
 	All edges stabilizers commute with $W_{f,f'}^Z(\rho)$ because they are of $Z$ type. 
 	Therefore all stabilizers of the toric code commute with $W_{f,f'}^Z(\rho)$. 
 	Hence $W_{f,f'}^Z(\rho)$ must be a stabilizer or a logical operator.  
 \end{proof}
 As we will see in Section~\ref{ssec:test-stab-lop}, 
 we need to know when a set of faces  support a  logical operator that occurs on $\mathcal{F}$.
 For this we need the following notion. 
 
 \begin{definition}[Face equivalence]
 	\label{def:face-eqiv}
 	We say two faces $f,f' \in \mathcal{F}$ are equivalent if there exist a face path $\rho$ between $f$ and $f'$ such that $W_{f,f'}^Z(\rho)$ is a stabilizer.
 \end{definition}
 
 \begin{figure}[H]
 	\centering
 	\subfigure[]
 	{ \includegraphics[scale=0.65]{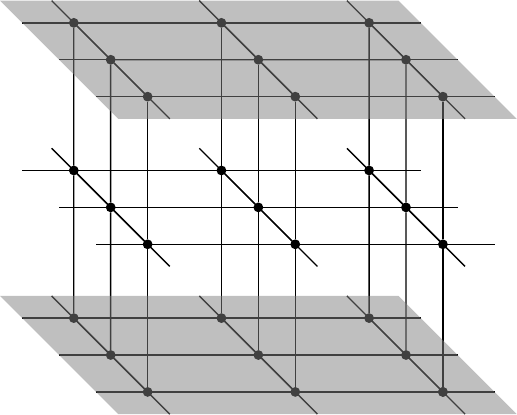}
 		\label{fig:deg-1-faces}}
 	\subfigure[]
 	{\includegraphics[scale=0.65]{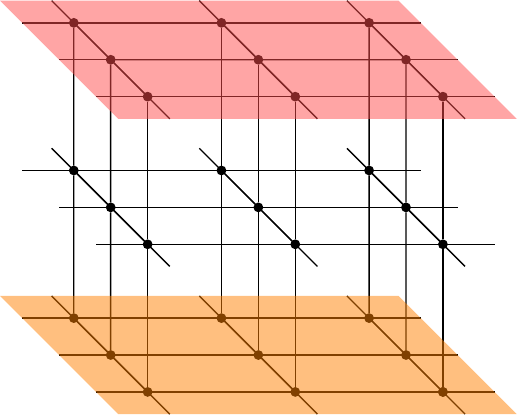}
 		\label{fig:deg-1-faces-split}}
 	\caption{(a) Set of all faces with degree one $\mathcal{F}$ (b) A partition of $\mathcal{F}$ into  two different equivalence classes. Each equivalence class is shown in a different color. Any two faces $f$, $f'$  of same color are connected by a face path $\rho$ such that $W_{f,f'}^Z(\rho)$, see Eq.~\eqref{eq:face-path-Z-op},  is a stabilizer. }
 	\label{fig:mathcalf}
 \end{figure}
 
 We assume that any face $f$ is equivalent to itself. 
 Any face path $\rho$ from face $f$ to $f'$ intersects at most twice in $\mathcal{F}$ i.e., $|\rho\cap \mathcal{F}|\leq 2$. 
 If $\rho$ intersects more than twice, then there exists a face $f_i \in \mathcal{F}$ for $ 1<i <m$. 
 Any such face is incident on two distinct volumes $\nu_{i-1}$ and $\nu_i$
 contradicting that $f_i$ is in $\mathcal{F}$.

 Face equivalence gives rise to a partition of $\mathcal{F}$.
 \begin{eqnarray}
 	\mathcal{F} = \bigcup_{j=1}^K \mathcal{F}_j, \label{eq:face-partition}
 \end{eqnarray}
 where $\mathcal{F}_i\cap \mathcal{F}_j=\emptyset$ for distinct $i$ and $ j$.
 Let $\{\mathcal{F}_j\}_{j=1}^{K}$ be the partition of $\mathcal{F}$ under face equivalence. 
 Then two faces $f,f' \in \mathcal{F}_i$ are equivalent.
 The partition also helps define the logical operators of  toric codes with boundaries. 
 First, we  show the following lemma relating the operator $W_{f,f'}^Z$ and logical operators.

 \begin{lemma}\label{lm:logical-ops-face-paths}
 	Let $f \in \mathcal{F}_i$ and $f' \in \mathcal{F}_j $ where  $i \neq j$. 
 	Then for any face path $\rho$ from $f$ to $f'$,
 	the operator $W_{f,f'}^Z(\rho)$ is a $Z$-type logical operator. 
 \end{lemma}
 \begin{proof}
 	By Lemma~\ref{lm:logop-stab-face-path}, we know that
 	$W_{f,f'}^Z(\rho)$ is a $Z$-type logical operator or a stabilizer. 
 	Suppose that it is not a logical operator. 
 	Then $f$ is equivalent to $f'$. 
 	In that case, $f$ and $f'$ belong to $\mathcal{F}_i$
 	contradicting that $f$ and $f'$ are in distinct 
 	equivalence classes $\mathcal{F}_i$ and $\mathcal{F}_j$.
 \end{proof}
 Next, we define the logical $X$ operators. 
 Given a set $T\subseteq C_{2}(\Gamma)$, let us define $X_T$ as follows.
 \begin{eqnarray}
 	X_T =\prod_{f\in T} X_f.\label{eq:Xop-set}
 \end{eqnarray}

 \begin{lemma}\label{lm:xtype-logical}
 	The operator $X_{\mathcal{F}_i}$ is an $X$ type logical operator. 
 \end{lemma}
 \begin{proof}
 	First we show that $X_{\mathcal{F}_i}$ commutes with all the $Z$ type stabilizers. 
 	In other words $X_{\mathcal{F}_i}$ commutes with $B_e$ for any edge $e$ in $C_1(\Gamma)$.
 	By Remark~\ref{rm:edgefacepath}, $\iota(e)$ is a face path or  a face cycle.
 	If $\iota(e) \cap \mathcal{F}_i = \emptyset $, then there is no overlap between $B_e$ and $X_{\mathcal{F}_i}$, therefore they commute.
 	Suppose not that $\iota(e)\cap \mathcal{F}_i\neq \emptyset$, then $\iota(e)$ is a face path but not a face cycle.
 	Further, the terminal faces of the face path should be from $\mathcal{F}$.
 	Since the intermediate faces of a face path are incident on two volumes only these terminal faces are from  $\mathcal{F}$.
 	Let these terminal faces be $f,f'$.
 	Now $f$ and $f'$ are face equivalent, because
 	the face path $\iota(e)$ from $f$ to $f'$ is exactly the  stabilizer $B_e$.
 	Therefore, 
 	$f, f'\in \mathcal{F}_i$, and $X_{\mathcal{F}_i}$
 	overlaps exactly twice with $B_e$ and it commutes
 	with $B_e$.
 	Thus $X_{\mathcal{F}_i}$ is a stabilizer or a logical operator.

 	Let $f\in \mathcal{F}_i$ and $f'\in \mathcal{F}_j$,
 	where $i\neq j$.
 	Consider a face path $\rho$ from $f$ to $f'$.
 	(By assumption all our graphs are connected
 	and there exists a face path between any pair of faces $f$ and $f'$.)
 	Note that the intermediate faces of a face path should be incident on at least two volumes. 
 	Therefore, none of the intermediate faces of 
 	the $Z$-type logical operator $W_{f,f'}^Z(\rho)$ 
 	are supported on $\mathcal{F}_i$ and only the faces
 	$f$ and $f'$ can be  supported by  $\mathcal{F}_i$.
 	By assumption $f'$ is not in $\mathcal{F}_i$, therefore, $W_{f,f'}^Z(\rho)$ has support on exactly one  face in $\mathcal{F}_i$.
 	Therefore, it anticommutes with $X_{\mathcal{F}_i}$,
 	which implies that  $X_{\mathcal{F}_i}$ cannot be stabilizer and must be a logical operator. 
 \end{proof}

 Note that for toric code with boundaries, any logical $Z$ type operator is equivalent to $W_{f,f'}^Z$ for some $f,f' \in \mathcal{F}$ and any $X$ type logical operator is equivalent to $X_{\mathcal{F}_i}$ for some $i$. 
 
 Here, we explain the breadth first approach to find a set of faces with unique solution in it.
 
 \subsection{Breadth first approach} 
 
 We use an algorithm similar to breadth first search to find the potential qubits in error. 
 This set $\mathcal{E}$ is first initialized as the empty set. 
 If there is a nonzero syndrome on an edge $e$, then there is at least one qubit incident on $e$ that is in error.
 Hence,  the faces $f$ incident on $e$, denoted $\iota(e)$,  are potential candidates  supporting the error. 
 We add one by one all the faces in $\iota(e)$ to $\mathcal{E}$ while ensuring that addition of $f$
 does not lead to a stabilizer or a logical operator in the
 support of $\mathcal{E}$.
 By requiring that $\mathcal{E}$ does not support stabilizers and logical operators, we ensure that that 
 $\mathcal{E}$ does not lead to multiple solutions. 
 
 In the next stage, we consider all the edges (one by one) which are incident on the qubits in $\mathcal{E}$ but have not yet been explored. 
 (These edges will carry zero syndrome.) 
 We consider adding the qubits  which are incident on this set of edges. 
 Again the addition of the qubits is such that the added qubits will not lead to a stabilizer or a logical operator within $\mathcal{E}$.
 This process is repeated for all faces in $\mathcal{E}$ which contain an edge which has not been explored yet. 
 That is all the qubits incident on that edge have not been considered for inclusion in $\mathcal{E}$.
 
 In principle we can stop when $\mathcal{E}$ supports an error which explains the measured syndrome. 
 However, testing for the existence of a solution at each step would increase the complexity of the algorithm. 
 Therefore, we do not test for the existence of a solution. 
 Instead, we continue to add more faces to $\mathcal{E}$ until all the faces have been explored. 
 At the end of this process, we would have explored all the qubits and $\mathcal{E}$ would be a set of potential qubits which can support exactly one error consistent with the syndrome.
 
 To summarize, we do a breadth first kind of search with nonzero syndromes as root.
 Instead of avoiding loops in  normal breadth first search, we avoid stabilizers and logical operators in the potential candidates. 
 To achieve this, we need to test the presence of support of stabilizer and/or logical operator in a given set of qubits.
 There are exponential number of stabilizers in the code.
 The number of logical operator to be searched is also exponential in the number of qubits. 
 This is because we have many exponential representatives for single logical operator.
 We make this exponential search into linear search in the following discussion.

 \subsection{Test for stabilizer and logical operator}\label{ssec:test-stab-lop}
 
 As we discussed earlier, we need to be able to identify if a given collection of faces support a stabilizer or a logical operator. 
 This can be characterized in topological terms.
 Consider an $X$ type stabilizer  generator $A_\nu$ attached to the volume $\nu$, see Eq.~\eqref{eq:x-stab}.
 The stabilizer  $A_\nu$ has its support on $\partial(\nu)$,
 the boundary of the volume $\nu$.
 This boundary divides the complex into two distinct volumes which are not connected to each other. 
 We call such a collection of faces a cut  set. 
 In general, a collection of faces is called a face cut set if their removal leads to two or more volumes which are not connected.  In the dual complex, the face path is a path and the face cut set is an edge  cut set. 
 A formal definition is  given below.

 \begin{definition}[Cut set]
 	A collection of faces $\mathsf{K}\subseteq C_2(\Gamma)$ is said to be cut set of $\Gamma$ if there exists two volumes $\nu$ and $\nu'$ such that for any face path  $\rho= \{f_1, f_2, \hdots f_m\}$ with $f_1 \in \partial(\nu)$, $f_m \in \partial(\nu')$ and $\nu,\nu' \notin \Lambda(\rho)$ , we have $\mathsf{K} \cap \rho \neq \emptyset$.
 \end{definition}
 
 \begin{figure}[H]
 	\centering
 	\includegraphics[scale=1]{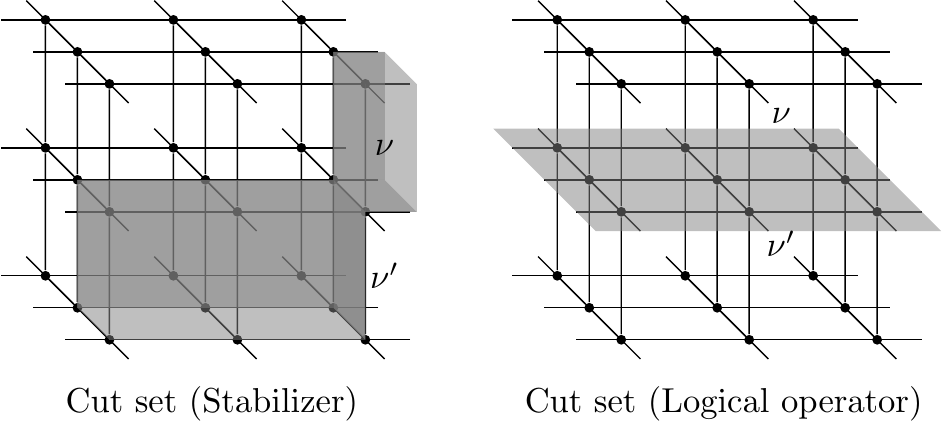}
 	\caption{
 		Two different cut sets. 
 		The shaded faces constitute the cut set. The volumes $\nu$ and $\nu'$  are separated by the cut set. The cut set on the left arise from the support,
 		while the one on the right from a logical operator.} 
 	\label{fig:cutset}
 \end{figure}

 Examples of cut sets arising from stabilizer generators and logical operators are shown in Fig~\ref{fig:cutset}.
 The collection of faces $F_c$ being a cut set implies that we have a stabilizer or logical operator with support within $F_c$. 
 
 This suggests that we can test for the presence of a stabilizer or a logical operator in a collection of faces by testing for the presence of a cut set in that collection. 
 The complexity of this test is linear in length of the code.
 
 However, not all stabilizers can be identified by a cut set.
 Consider the stabilizer $\prod_{\nu \in C_3} A_{\nu}$.
 The support of this stabilizer is given by $\partial(C_3)$
 and it is the set of all faces which are incident on only one volume.  
 The support of this stabilizer encloses all the volumes and it is not a cut set, see Fig.~\ref{fig:deg-1-faces}.

 There are also some logical operators whose support is not cut set. 
 Therefore, we need to be able to determine if a stabilizer or a logical operator is support in $\mathcal{E}$ even when they do correspond to a cut set. 
 
 Observe that in Fig.~\ref{fig:mathcalf} when the support of a stabilizer or a logical operator is not a cut set, then the support is a subset of  $\mathcal{F}$.
 This is true in general. 
 For instance, consider a stabilizer whose support is not a cut set.
 If there is at least one face in the support which is incident on two different volumes $\nu$ and $\nu'$, then face is a cut set for $\nu$ and $\nu'$.
 
 \begin{lemma}[]\label{lm:stab-not-cutset}
 	Suppose that $M$ is  an $X$ type stabilizer or a logical operator whose support is not a cutset. 
 	Then each face in the support of $M$ is incident exactly on one 3-cell. In other words, 
 	$\supp(M)\subseteq \mathcal{F}$.
 \end{lemma}
 \begin{proof}
 	
 	First we show that if a support of stabilizer $M$ is not a cutset, then $\supp{M} \subseteq \mathcal{F}$.
 	If $M$ is a stabilizer, then it can be written as a product of volume stabilizers $A_\nu$. 
 	The support of $M$ is $\partial(A)$.
 	If $A \neq C_3(\Gamma)$, then the support of $M$ forms a 
 	cutset for a volume $\nu\in A$ and $\nu'\not\in A$ contradicting that $\supp(M)$ is not a cutset.
 	Therefore $M= \prod_{\nu \in C_3(\Gamma)} A_{\nu}$. 
 	Suppose a face $f$ is incident on two volumes $\nu$
 	and $\nu'$. 
 	Then $A_\nu A_{\nu'}$ has no support on $f$.
 	Therefore $M$ does not have support of faces which are incident on two volumes. 
 	Alternatively, $\supp(M)\subseteq \mathcal{F}$.
 	
 	Now, we show that if $M$ is a logical operator for which $\supp{M}$ is not a cutset, then $\supp{M} \subseteq \mathcal{F}$.
 	Suppose that $M$ is a logical operator with $\supp(M)$ not a cutset. 
 	If $\supp(M)\subseteq \mathcal{F}$, then the lemma holds. 
 	Assume therefore, that  
 	$ \supp(M)\setminus \mathcal{F}\neq \emptyset$. 
 	Suppose $f\in \supp(M)\setminus \mathcal{F}$, then $f$ is incident on two volumes $\nu$ and $\nu'$. 
 	Since, $\supp(M)$ is not a cutset, there exists some 
 	face path $\rho$ from $\nu$ to $\nu'$
 	such that $\rho\cap \supp(M)=\emptyset$.
 	So $W_{\nu,\nu'}^Z(\rho)$ commutes with $M$.
 	Now $\rho \cup  \{ f\}$ is a face cycle. 
 	Therefore $W_{\nu,\nu'}^Z(\rho) Z_f$ commutes with $M$. 
 	But since $f\in \supp{M}$, $Z_f$ anticommutes with $M$. Hence $W_{\nu,\nu'}^Z$ should anticommutes with $M$ which contradicts our assumption.
 \end{proof}
 
 Now we explain how we detect if such stabilizers and logical operators  are supported in a collection of faces. 
 We add dummy volumes on each boundary and construct a new lattice $\tilde{\Gamma}$ in  which the support of every 
 stabilizer or logical operator is a cut set.
 
 Now, we can construct $\tilde{\Gamma}$ by adding new volumes to $\Gamma$.
 Since $\mathcal{F}_i$ is the support of a logical operator which is not a cut set and the support of the stabilizer which is not a cut set is $\mathcal{F}$, it is enough to construct new lattice such that  $\mathcal{F}_i$ are cut sets.
 We can add new volumes $\bar{\nu}_i$ for each $\mathcal{F}_i$ such that its boundary is $\mathcal{F}_i$. That is $\partial(\bar{\nu}_i)= \mathcal{F}_i$. Fig.~\ref{fig:tildegamma} shows an example for new volumes and construction of $\tilde{\Gamma}$. 
 \begin{figure}
 	\centering
 	\subfigure[]
 	{\includegraphics[scale=0.75]{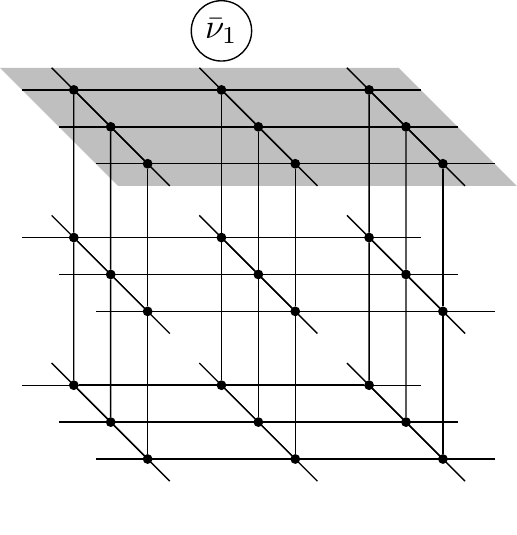}}
 	\subfigure[]
 	{\includegraphics[scale=0.75]{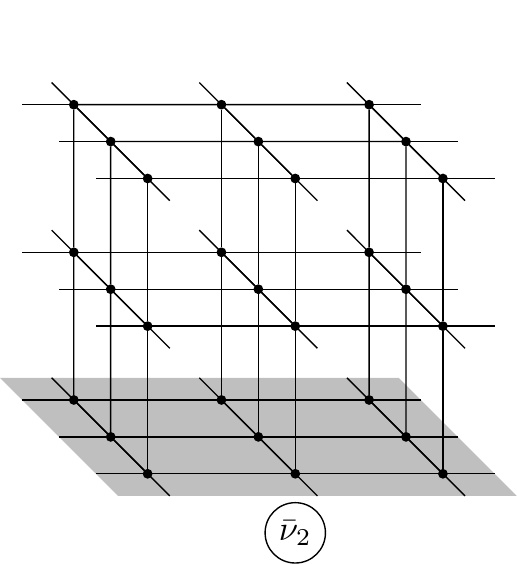}}
 	\subfigure[]
 	{\includegraphics[scale=0.75]{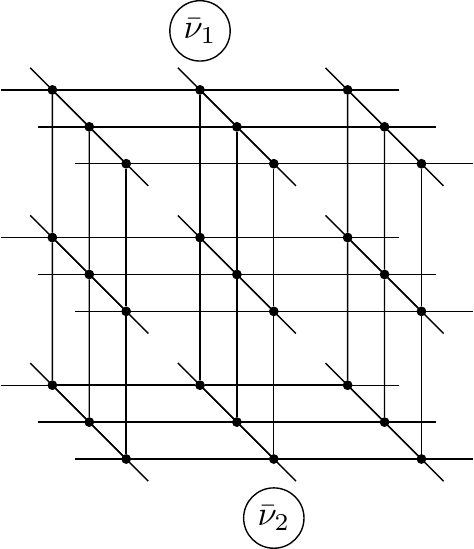}}
 	\caption{(a) An artificial volume $\bar{\nu_1}$ is added adjacent to all the shaded faces (b) Another artificial volume  $\bar{\nu_2}$ is added below the shaded faces. (c) Final augmented lattice $\tilde{\Gamma}$ with all the artificial volumes. One volume per equivalance class of $\mathcal{F}$.}
 	\label{fig:tildegamma}
 \end{figure}
 
 From the construction of $\tilde{\Gamma}$,
 we see that every face in $\tilde{\Gamma}$ is incident on two volumes, at most one of which could be a dummy volume.

 Now we can identify the logical operator or stabilizer using the cut set in $\tilde{\Gamma}$. The following lemma formalizes this.
 \begin{lemma}[Cut sets in augmented lattice]\label{lm:cutset-E}
 	For a 3D toric code with boundaries on $\Gamma$, 
 	a set of faces $\mathcal{E}$ is a cut set in $\tilde{\Gamma}$ if and only if there exists an
 	nontrivial $X$ type logical operator or stabilizer with support in  $\mathcal{E}$. 
 \end{lemma}
 \begin{proof}
 	In $\tilde{\Gamma}$ the support of any $X$ stabilizer or logical operator is the boundary of a collection of volumes $V $. 
 	Observe that $V\subsetneq C_3(\tilde{\Gamma})$ because $X_{\partial(V)}=I$ when $V=C_3(\tilde{\Gamma})$.
 	Consider any face path $\rho$ from a volume $\nu$ from $V$ and another volume $\nu'$ from $C_3(\tilde{\Gamma}) \setminus V$.  
 	Let us consider a face path $f_1,\hdots, f_m$ from $\nu$ to $\nu'$ and none of $f_i$ is from $\partial(V)$.
 	Since $f_1 \notin \partial(V)$ and $\nu = \nu_0 \in V$, we have $\nu_1 \in V$. 
 	In general, if $\nu_i \in V$ and $f_{i+1} \notin \partial(V)$, then $\nu_{i+1} \in V$. Thus we get $\nu' \in V$ because all $f_i \notin \partial(V)$. This contradicts our assumption that $\nu' \notin V$. 
 	Hence any stabilizer or logical operator support will be a cut set.
 	
 	Now consider  any cutset $F_c$.
 	Pick any volume $\nu$. 
 	Let $V$ be the collection of volumes each of which is connected to $\nu$ by a face path $\rho$
 	such that $\rho\cap F_c=\emptyset$.
 	This implies $\partial(V) \subset F_c$.
 	The set $V \neq C_3(\tilde{\Gamma})$ because $F_c$ is a cut set.  
 	$\prod_{v \in V} X_{\partial(v)}$ is a logical operator or stabilizer because each $X_{\partial(v)}$ is a logical operator or stabilizer.
 	Hence for any cutset $F_c$, there exists a stabilizer with support within $F_c$.
 \end{proof}
 Lemma~\ref{lm:cutset-E} will enable  us  to check if a stabilizer or logical operator is supported in a
 given set of faces. 
 We summarize the algorithm to find the potential support of error in the follwoing algorithm \ref{alg:tcc-bound-potqubits},
 \begin{algorithm}[H]
 	\caption{Identifying potential qubits in error for toric codes with boundaries} \label{alg:tcc-bound-potqubits}
 	\begin{algorithmic}[1]
 		\REQUIRE A 3D complex $\Gamma$, collection of edges $S_E$ carrying nonzero syndrome.
 		\ENSURE Collection of faces $\mathcal{E} \subseteq C_2(\Gamma)$ such that there exists $ \mathcal{E}' \subseteq \mathcal{E}$ and $ \partial (\mathcal{E}') = S_E$
 		\STATE Construct $\tilde{\Gamma}$ by adding dummy volumes  $\bar{\nu}_i$ with $\partial(\bar{\nu}_i) = \mathcal{F}_i$ where $\mathcal{F}_i$ is a equivalence class obtained from Eq.~\eqref{eq:face-partition}.
 		
 		\STATE Initialize the boundary as $B=S_E$, $\mathcal{E}= \emptyset$ and mark all faces as unexplored.
 		\WHILE{there exist unexplored faces}
 		\STATE	$B'=\emptyset$ \COMMENT{Boundary for next stage}
 		\FOR{all unexplored faces $f$ incident on the boundary $B$}
 		\STATE Mark the face $f$ as explored.
 		\IF{ the set $\mathcal{E} \cup \{f\} $ is not a cut set in $\tilde{\Gamma}$} 
 		\STATE Update $\mathcal{E}=\mathcal{E}\cup \{ f\}$ \COMMENT{$f$ is a potential qubit in error}
 		\STATE Update $B'=B'\cup \partial(f)$ 
 		\ENDIF
 		\ENDFOR
 		\STATE $B=B'\setminus B$ 
 		\ENDWHILE
 		\STATE Return $\mathcal{E}$ and exit
 	\end{algorithmic}
 \end{algorithm}
 
 In line~7 of Algorithm~\ref{alg:tcc-nobound-potqubits} we need a method to test whether a set of faces $\mathcal{E}$ is a cut set or not. 
 If $\mathcal{E}$ is a cutset, then every volume 
 $\nu$ is reachable from any other volume $\nu'$ through faces not in $\mathcal{E}$.
 The complexity of this step is $O(|C_2(\Gamma)|)$.
 In other words it is linear in the number of qubits. It is same as the error trapping algorithm used for detecting closed volumes in \cite{kulkarni18}.
 
 The set of faces returned by 
 Algorithm~\ref{alg:tcc-bound-potqubits} satisfy the 
 following properties. 
 
 \begin{lemma}\label{lm:prop-induced-erasures}
 	Let $\mathcal{E}$ be the set of faces returned by   Algorithm~\ref{alg:tcc-bound-potqubits}. Then the following properties hold
 	\begin{compactenum}[(a)]
 		\item If $A$ is a nonempty collection of faces such that $X_{A}$ is a stabilizer or a logical operator, then $A \nsubseteq \mathcal{E}$.
 		\item If any face $f \notin \mathcal{E}$, 
 		then $\mathcal{E}\cup \{ f \}$ supports a stabilizer or a logical operator $X_{B}$
 		where $B\subseteq \mathcal{E}\cup \{ f \}$.
 	\end{compactenum}
 \end{lemma}
 \begin{proof}
 	First let us prove (a). Assume $X_A$ be a stabilizer and $A \subset \mathcal{E}$. Let $f$ be the last face added to $\mathcal{E}$ from the set $A$.  
 	While adding $f$, let the set of erasures be $\mathcal{E}'$. Hence we have $\mathcal{E}' \subset \mathcal{E}$.
 	The face $f$ will not be added to $\mathcal{E}'$ because $f \cup \mathcal{E}'$ is a cutset by 
 	Lemma~\ref{lm:cutset-E}. 
 	Thus $f$ will not be present in $\mathcal{E}$. 
 	Therefore we cannot have a stabilizer of a logical operator $X_A$ with $A \subset \mathcal{E}$.

 	Next we prove (b).
 	From Algorithm \ref{alg:tcc-bound-potqubits}, we do not add a face $f$ to $\mathcal{E}$, only if $f$ along with a subset of $\mathcal{E}$ (set of faces in $\mathcal{E}$ while checking the cutset for $ \mathcal{E} \cup \{f\}$) is a cutset. This implies that for any face $f \notin \mathcal{E}$, $f \cup \mathcal{E}$ is a cutset. This proves (b). 
 \end{proof} 
 
 \subsection{Finding the unique solution}
 From Algorithm \ref{alg:tcc-bound-potqubits}, we obtain a collection of faces $\mathcal{E}$ which
 support an error $E$ 
 such that $\partial E$ is the nonzero syndrome. 
 Now we show how to find $\supp(E)$.
 First, we show that the existence of a solution within $\mathcal{E}$ and then show that it is unique. 
 
 \begin{theorem}[]
 	\label{lm:peel-bound}
 	Let $S_E$ be the support of nonzero syndrome input  and $\mathcal{E}$ be the set of faces returned by Algorithm~\ref{alg:tcc-bound-potqubits}. 
 	Then, there exist a unique set of faces $\hat{\mathcal{E}} \subseteq \mathcal{E} $ such that $\partial(\hat{\mathcal{E}}) = S_E $.
 \end{theorem}
 \begin{proof}

 	Suppose that $\mathcal{E}'$ is a set of faces 
 	such that $\partial\mathcal{E}'$ is $s_E$ the boundary observed.
 	If $\mathcal{E}'$ is not a subset of $\mathcal{E}$,
 	then there exists some $f$ in 
 	$\mathcal{E}'\setminus \mathcal{E}$.
 	Since this face was not added to $\mathcal{E}$,
 	then it means that $\mathcal{E}\cup \{ f \}$
 	supports a stabilizer or a logical operator. 
 	Let the support of this operator be $A\cup \{ f\}$
 	Thus the boundary of $f$ is same as the boundary of 
 	$A$.
 	We can replace $f$ in $\mathcal{E}'$ by $A\subseteq \mathcal{E}$.
 	Thus we can obtain another set $\mathcal{E}''$
 	which does not have the face $f$.
 	We repeat this process with $\mathcal{E}''$
 	until its support is entirely in $\mathcal{E}$.
 	At which point we have a set $\hat{\mathcal{E}}\subseteq \mathcal{E}$ whose boundary is the observed syndrome $S_E$.
 	This shows that there exists a solution in $\mathcal{E}$.
 	
 	If there exists another solution $\hat{\mathcal{E}}' \subseteq F$ such that $\hat{\mathcal{E}}'=  S_E$, then $X_{\hat{\mathcal{E}}}X_{\hat{\mathcal{E}}'}$
 	has zero syndrome. 
 	Therefore, it must be 
 	a stabilizer or a logical operator with support $ \hat{\mathcal{E}} \triangle \hat{\mathcal{E}}'$, where 
 	$\triangle$ is the symmetric difference of sets. Now, $\hat{\mathcal{E}} \triangle \hat{\mathcal{E}}'  \subset \mathcal{E} $.
 	However, this contradicts Lemma~\ref{lm:prop-induced-erasures}  which claims that $\mathcal{E} $ does not support a stabilizer or logical operator.
 \end{proof}
 
 Now finding the unique solution from the available support is the task to be solved.
 One can think of the set $\mathcal{E}$ as a set of erasures and sovle a system of linear equations. 
 Alternatively, we can estimate this iteratively
 using the peeling decoder used for classical codes,
 for instance see \cite{kulkarni18,delfosse17a}. 
 For completeness, we give the algorithm 
 adapted to our perspective, see Algorithm~\ref{alg:peeling-bound}.
 
 In general, when we using peeling to correct for erasures, it can fail to give a unique solution when the set of erasures supports a stabilizer or logical operator. 
 Peeling also fails in the sense, it cannot proceed further, when every check involves two or more erased qubits.
 This problem can lead to a decoding failure in our algorithm as well. 
 
 A third case of decoder failure which occur due to the presence of Klein bottle-like  structure 
 as was observed in \cite{kulkarni18,duivenvoorden18}.
 We conjecture that the occurrence of such patterns will be rare because we do not consider an arbitrary erasure pattern but one which is initiated from the nonzero syndromes. 
 
 Since the algorithm finds the nearest faces first, we claim that Klein bottle-like structure is a rare case which happens only for specific syndrome patterns.
 In such cases we can either solve the system of linear equations or declare decode failure.
 Another way to solve this is to freeze an arbitrary face to be not in error and then continue peeling.
 There is a chance that we have frozen the actual error to be not in error which does not clear the non zero syndromes.
 We can repeat the algorithm to clear such cases or declare decoder failure.
 In our case, we have declared decoder failure.
 
 \begin{algorithm}[H]
 	\caption{Peeling for toric codes with boundaries \cite{kulkarni18}} \label{alg:peeling-bound}
 	\begin{algorithmic}[1]
 		\REQUIRE A 3D lattice $\Gamma$, a collection of edges $S_E$ which is the support of non zero syndrome and a collection of faces of possible error positions $\mathcal{E}$.
 		\ENSURE A collection of faces  $\hat{\mathcal{E}} \subset \mathcal{E}$ such that $\partial(\hat{\mathcal{E}}) = S_E$.
 		\COMMENT{$X_{\hat{\mathcal{E}}}$ gives the same syndrome as $S_E$.}
 		\STATE $B=S_E$ \COMMENT{Nonzero syndromes}
 		\STATE Initialize the error estimate support $\hat{\mathcal{E}} =\emptyset$.
 		\WHILE{there exists an edge $e$ in $\Gamma$ which is incident on exactly one face $f$ in $\mathcal{E}$} 
 		\IF{ the  edge $e$  is in  $B$}
 		\STATE Update  $\hat{\mathcal{E}}= \hat{\mathcal{E}}\cup\{ f\} $ 
 		\STATE Update the nonzero syndrome $B =B \triangle \partial(f)$
 		\ENDIF
 		\STATE Update $\mathcal{E}= \mathcal{E}\setminus \{f \}$ 
 		\ENDWHILE
 		\IF{ $B = \emptyset $}\COMMENT{All nonzero syndromes are cleared }
 		\STATE Return $\hat{\mathcal{E}}$.
 		\ELSE
 		\STATE Return Decoder failure
 		\ENDIF
 	\end{algorithmic}
 \end{algorithm}

 \subsection{Putting all the pieces together}
 
 Now, we summarize the decoding algorithm for toric codes with boundaries.  Given the syndromes, first we find a set of faces $\mathcal{E}$ such that there exists a unique solution within $\mathcal{E}$ explaining the observed syndrome. 
 Then we treat the set $\mathcal{E}$ as  erasures and estimate the unique solution using erasure decoding algorithm. 
 This method is summarized in 
 Algorithm~\ref{alg:ovallbound}.
 \begin{algorithm}[H]
 	\caption{Decoding 3D toric codes with boundaries} \label{alg:ovallbound}
 	\begin{algorithmic}[1]
 		\REQUIRE 3D lattice $\Gamma$, collection of edges $S_E$ with nonzero syndrome.
 		\ENSURE Collection of faces $\hat{\mathcal{E}}$ such that $  \partial(\hat{\mathcal{E}})=E$
 		\STATE Find potential qubits in error using 
 		Algorithm~\ref{alg:tcc-bound-potqubits} with $S_E$ as input and obtain $\mathcal{E}$. 
 		\STATE Estimate $\hat{\mathcal{E}}$  using Algorithm~\ref{alg:peeling-bound} with $\mathcal{E}$ and $S_E$ as inputs.
 		\IF{ $\hat{\mathcal{E}}\neq \emptyset$ } 
 		\STATE Return $\hat{\mathcal{E}}$ and exit.
 		\ELSE
 		\STATE Report decoder failure and exit.
 		\ENDIF
 	\end{algorithmic}
 \end{algorithm}
 
 This concludes our discussion on toric codes with boundaries.
 We next study 3D codes without boundaries. 
 
 \section{Decoding 3D toric codes with periodic boundaries} 
 In this section, we consider decoding bit flip errors on 3D toric codes on lattices with periodic boundary conditions. 
 For clarity of presentation, we restrict our attention to the case 
 when the code is periodic in all three directions.
 Such a code encodes three logical qubits. 
 
 Recall that in Section~\ref{sec:intuition},
 we had identified four pieces in the design of the decoding algorithm. 
 We begin by exploring the neighbourhood of the nonzero syndromes to try and explain the syndrome that we observed. 
 Then we try to freeze errors on certain qubits.
 More precisely, we freeze a qubit $f$
 if it along with $\mathcal{E}$, the potential qubits that have been explored, supports a stabilizer or a logical operator. 
 This required us to test for the presence of a cutset in 
 $\mathcal{E}\cup \{ f\}$.
 This is the step that fails for codes with periodic boundary conditions. 
 It only fails for certain nontrivial logical $X$ operators and not the $X$ stabilizers. 
 Consider the logical operator shown in Fig.~\ref{fig:3dtclogopnocutset}.
 The support of this  operator is not a cut set. 
 Since any two volumes $\nu$ and $\nu'$ can be connected by a face path which no support in the logical operator. 
 Similarly, all the three $X$ logical operators there exists one or more equivalent logical operators which do not correspond to a cut set.  
 
 \begin{figure}[H]
 	\centering
 	\subfigure[\label{fig:3dtclogopnocutset}]
 	{\includegraphics[scale=0.85]{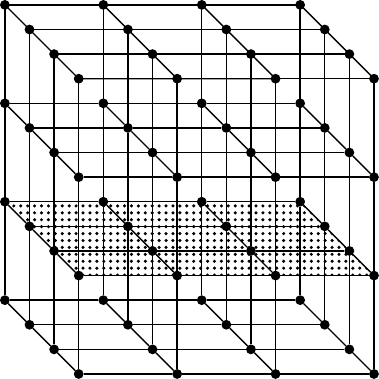}}~
 	\subfigure[\label{fig:3dtclogop}]
 	{\includegraphics[scale=0.85]{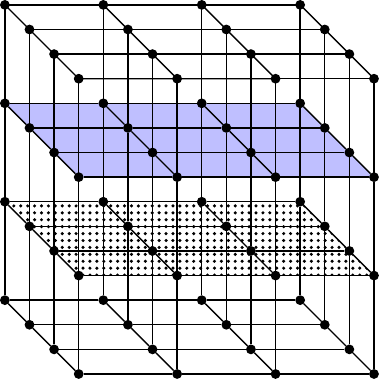}}
 	\caption{(a) An $X$ logical operator $\mathcal{L}$ which is not a cutset.
 		Any pair of volumes $\nu$ and $\nu'$ are connected by a face path not intersecting with $\supp(\mathcal{L})$.
 		(b) By adding an artificial boundary $\mathcal{X}$ (blue faces)
 		the support of the logical operator becomes a cut set along with $\mathcal{X}$. 
 		Observe now that any face path between $\nu$ and $\nu'$ intersects with  $\mathcal{X}$ or the support of $\mathcal{L}$.}
 	\label{fig:logopnocutset3dwithoutbound}
 \end{figure}
 
 To overcome this problem, we could add artificial boundaries ensuring that the support of every logical operator is a cutset.
 In Fig.~\ref{fig:logopnocutset3dwithoutbound}, we illustrate 
 for the logical operator $X_L$ where $L=\supp(X_L)$.
 Here, we 
 have added an artificial boundary $\mathcal{X}$. 
 Note that $\mathcal{X}$ is not a cut set. This implies that $\mathcal{X}$
 cannot contain the support of a stabilizer. 
 Now, the support of the logical operator $X_L$ and $\mathcal{X}$, i.e. $L\cup \mathcal{X}$, is a cut set. 
 With this modification, the supports of $X$ stabilizers are still cut sets as before.
 
 However, there remain other logical operators which do not correspond cut sets with $\mathcal{X}$. 
 This is illustrated in in Fig.~\ref{fig:anotherlogopnocutset}.
 
 \begin{figure}[H]
 	\centering
 	\includegraphics[scale=0.85]{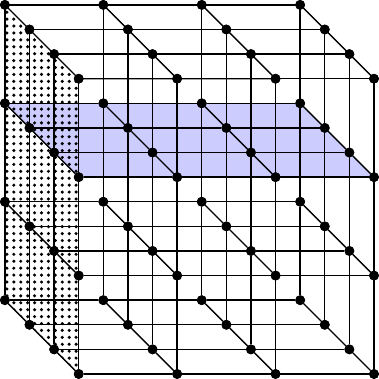}
 	\caption{Figure shows another logical operator which is not a cut set with an artificial boundary shown in 
 		Fig.~\ref{fig:3dtclogopnocutset}.}
 	\label{fig:anotherlogopnocutset}
 \end{figure}
 
 So we  extend the  artificial boundary. 
 Specifically, we take this to be the union of the supports of all the  independent $X$ type logical operators.
 In other words, 
 \begin{eqnarray}
 	\mathcal{X} =\cup_{i}\supp(\overline{X}_i)\label{eq:chi-first-cut}
 \end{eqnarray}
 This is illustrated in Fig.~\ref{fig:chi}.
 However, this is not yet a complete solution. 
 For a logical operator $X_L$, where  
 $L\subseteq \mathcal{X}$, then $L\cup \mathcal{X}$ will not be a cut set.  
 We avoid this situation by not considering the qubits in $\mathcal{X}$ while exploring.
 In other words, we do not add qubits in $\mathcal{X}$ to 
 $\mathcal{E}$.
 What this ensures is that if a solution exists in $\mathcal{E}$,
 it will be unique. 
 
 However, the algorithm is not yet complete. 
 For instance, if the original error is restricted to $\mathcal{X}$, then $\mathcal{E}$ does not contain a solution. 
 We address this by decoding on $\mathcal{X}$ separately at the end. 
 We will show that this resolves any other problems created due to the addition of boundaries and/or not considering $\mathcal{X}$
 in the initial exploration.

 We use the toric code defined on arbitrary lattice and encodes three logical operators. 
 Logical $X$ corresponds to  non trivial plane and logical $Z$ operators corresponds to a non trivial face cycle in all three directions. 
 Test for a logical operator is not same as that of codes with boundaries.
 The same procedure of breadth first search approach can be used to find the potential qubits in error.
 But to test for stabilizer and logical operator we use a different approach which we discuss in the next section.

 \subsection{Finding a potential set of qubits that can explain the syndrome}
 For toric codes with periodic boundaries, the 
 augmented lattice $\tilde{\Gamma}$ is same as $ \Gamma$ because of the absence of the faces with degree one. 
 In case of toric codes with boundaries, the  stabilizer
 $\prod_{\nu \in C_3} A_{\nu}$ did not correspond to a cut set in 
 $\Gamma$. 
 So we had to test for it separately.
 In the present case, $\prod_{\nu \in C_3} A_{\nu}=I$
 and it does not to be tested separately. 
 All stabilizers correspond to  cut sets.

 Since the stabilizer $\prod_{\nu \in C_3} A_{\nu} = I$, we can have all stabilizer identified by a cut set. 
 Formally we can state this in the following lemma.
 \begin{lemma}[Stabilizer cut set]
 	\label{lm:stabcutcettccnobound}
 	For 3D toric code on a lattice $\Gamma$ with periodic boundary conditions, the collection of faces $F_c$ is a cut-set for $\Gamma$ if and only if a nontrivial stabilizer is supported in   $F_c$.
 \end{lemma}
 \begin{proof}
 	Since any nontrivial stabilizer can be written $X_{\partial(V)} = \prod_{\nu \in V} X_{\partial(\nu)}$ for some $V \subsetneq C_3(\Gamma)$.
 	$V \neq C_3(\Gamma)$. 
 	So for any two volumes $\nu$ from $V$ and $\nu'$ from $C_3(\Gamma) \setminus V$, any face path $\rho$ between $\nu$
 	and $\nu'$ must have nonempty overlap with $\partial(V)$. 
 	Hence support of a stabilizer is a cutset.
 	
 	If $F_c$ is a cutset in $\Gamma$, then we can get a collection of volumes $V$ such that $\partial(V) \subset F_c$. For a collection of volumes $V$ in $\Gamma$, we have $\partial(V)$ is a support of stabilizer.
 	Hence we have a support of stabilizer $\partial(V) \subset F_c$.
 \end{proof}

 As mentioned some 
 logical operators of these toric codes do not form a cut set.
 To make the support of every logical operator a cut set, 
 we shall add artificial boundaries.
 We formally explain this concept of artificial boundary in the following discussion.
 Denote by $\mathcal{L}$ the set of logical qubits encoded by the toric code.
 \begin{definition}[Artificial boundary] \label{def:artbound}
 	For a 3D toric code, the artificial boundary $\mathcal{X}$ is defined as
 	as a set of faces in $C_2(\Gamma)$ that contains the support of all the logical $X$-type operators but does not contain the support of a stabilizer. 
 	\begin{eqnarray}
 		\mathcal{X} = \bigcup\limits_{i \in \mathcal{L}} \supp(\bar{X}_i^r),\label{eq:chi}
 	\end{eqnarray}
 	where $\bar{X}_i^r$ is a representative of $\bar{X_i}$.
 \end{definition}
 
 If we allow $\mathcal{X}$ to contain the support of a stabilizer, then for all faces $f\in C_2(\Gamma)$ we have $\mathcal{X}\cup \{ f\}$ as a cut set and no face will be added to $\mathcal{E}$.
 Thus $\mathcal{E}$  will remain empty and the decoding algorithm will not work. . 
 For this reason we do not include the support of a stabilizer in 
 $\mathcal{X}$.
 The existence of such a $\mathcal{X}$ can be shown as follows. 
 
 \begin{lemma}[Existence of artificial boundary]\label{lm:chi-existence}
 	The boundary $\mathcal{X}$ as defined in Eq.~\eqref{eq:chi}
 	exists. 
 \end{lemma}
 \begin{proof}
 	
 	Consider the set $\mathcal{X}' = \bigcup\limits_{i \in \mathcal{L}} \supp(\bar{X_i}^{r'})$ such that there exists a stabilizer $X_S$ with support on $S \subset \mathcal{X}$. 
 	We claim that we can always construct an artificial boundary $\mathcal{X}$ based on Definition \ref{def:artbound}. 
 	
 	Take a logical operator $\bar{X}_j^{r'}$ such that 
 	$T = \supp(\bar{X_j}^{r'}) \cap S \neq \emptyset $. 
 	Starting from $i=1$, for all logical operators if $\supp(\bar{X_i}^{r'}) \cap T \neq \emptyset $, then update 
 	$T = \supp(\bar{X_i}^{r'}) \cap T$.
 	Repeat this for all independent logical operators $\bar{X_i}$ in a sequence. 
 	Now, for each logical operators $\bar{X_i}$, if $\supp(\bar{X_i}^{r'}) \cap T \neq \emptyset$, then set the new representative $\bar{X_i}^r$ as $\bar{X_i}^{r'} X_S$ otherwise $\bar{X_i}^r = \bar{X_i}^{r'}$.  
 	Then the new operators $\bar{X}_i^r$ has support in $\mathcal{X}'$ but none of them are supported on $T$.
 	Thus,  $\mathcal{X} = \bigcup\limits_{i \in \mathcal{L}} \supp(\bar{X_i}^r)$ is a proper subset of $\mathcal{X}'$ and not supported in $T$. 
 	Since $T\subseteq S$ and $T\nsubseteq \mathcal{X} $,
 	it follows that $\mathcal{X}$ does not contain the stabilizer 
 	$X_S$.
 	We can repeat this process till all stabilizers in the support of the artificial boundary are removed. 
 	Thus we can always find an artificial boundary $\mathcal{X}$.
 	as defined in 
 	definition~\ref{def:artbound}.
 \end{proof}
 
 Note that this artificial boundary is not unique. 
 See Fig.~\ref{fig:chi} for examples of artificial boundary satisfying the conditions in Definition~\ref{def:artbound}.

 \begin{figure}
 	\centering
 	\subfigure[\label{fig:chione}]{\includegraphics[scale=0.8]{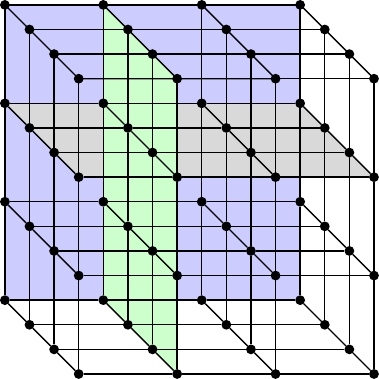}}~
 	\subfigure[\label{fig:chitwo}]{\includegraphics[scale=0.8]{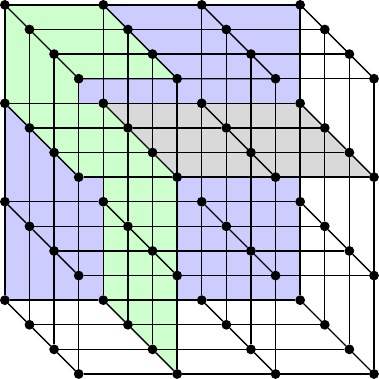}}
 	\caption{Two examples of $\mathcal{X}$ for a toric code which encodes 3 logical qubits. Observe that $\mathcal{X}$ is not a cutset in $\Gamma$. The shaded faces shows the faces in $\mathcal{X}$.}
 	\label{fig:chi}
 \end{figure}
 
 Also for any artificial boundary $ \mathcal{X}$ and for any single qubit logical operator $\bar{X_i}$ , we have an equivalent operator $\bar{X}_i$ such that $\supp{\bar{X_i}^r} \subset \mathcal{X}$ and $\bar{X_i} = \bar{X_i}^r X_S$ where $X_S$ is a $X$-type stabilizer with support $S$. 
 For any set of faces $F$ , if there exists a logical operator supported by $F$ , then $F$ combined with $\mathcal{X}$ will support a stabilizer.

 Recall that using Lemma~\ref{lm:stabcutcettccnobound} we can check for the presence of a stabilizer in a subset of faces $\mathcal{E}$.  
 Now we  check for the presence of 
 a logical operators in $\mathcal{E}$ by testing for the presence of a cut set in $\mathcal{E}\cup \mathcal{X}$. 
 The following lemma  proves this formally.
 
 \begin{lemma}[Detecting logical operators outside $\mathcal{X}$]
 	\label{lm:tcc-logop-nobound}
 	Let $\Gamma$ be a 3D toric code and  $\mathcal{X}$ be an artificial boundary as in Eq.~\eqref{eq:chi} and $\mathcal{E}$ be some collection of faces such that $\mathcal{E} \cap \mathcal{X} = \emptyset$. Let $\bar{X_i}$ be any logical operator such that $\supp{\bar{X_i}} \subseteq \mathcal{E}$. Then, $\mathcal{E} \cup \mathcal{X}$ is a cutset in $\Gamma$.
 \end{lemma}
 \begin{proof}
 	From the definition of $\mathcal{X}$ as given in Eq.~\eqref{eq:chi}, we see that there exists an $X$
 	logical operator $\tilde{X}_i $ that is  equivalent to $\bar{X}_i$ which is supported in $\mathcal{X}$. 
 	The equivalence of $\bar{X}_i$ and $\tilde{X}_i $ 
 	implies that $X_S  =  \tilde{X}_i \bar{X}_i $
 	is a stabilizer. 
 	Further, $S \subseteq \supp(\bar{X}_i) \cup \supp(\tilde{X}_i)$.
 	By assumption, $ \supp(\bar{X}_i) \subseteq \mathcal{E}$ and 
 	$\supp(\tilde{X}_i)\subseteq \mathcal{X}$ and 
 	Eq.~\eqref{eq:stab-supp} follows. 
 	\begin{eqnarray}
 		S & \subseteq  \mathcal{X} \cup \mathcal{E} \label{eq:stab-supp}
 	\end{eqnarray}
 	From Lemma \ref{lm:stabcutcettccnobound} and Eq.~\eqref{eq:stab-supp}, $\mathcal{X} \cup \mathcal{E}$ is a cutset in $\Gamma$. 
 \end{proof}
 
 For a logical operator with support within $\mathcal{X}$, Lemma~\ref{lm:tcc-logop-nobound} cannot be used. 
 The proof assumes that the combination of logical operators will yield a nontrivial stabilizer which is not true if the logical operator is entirely in $\mathcal{X}$.
 
 From Lemmas \ref{lm:stabcutcettccnobound}~and~\ref{lm:tcc-logop-nobound}, by checking for cut set in  $\mathcal{X} \cup \mathcal{E}$, we can identify and remove the presence of logical operators and stabilizers from $\mathcal{E}$ except for logical operators within $\mathcal{X}$. 
 Therefore, adding the qubits in $\mathcal{X}$ to $\mathcal{E}$
 would lead to multiple solutions in $\mathcal{E}$.
 To avoid this we do not add faces from $\mathcal{X}$ to 
 $\mathcal{E}$.
 We correct the errors in $\mathcal{X}$ separately.
 We give the algorithm for clearing the errors in $C_2(\Gamma)\setminus \mathcal{X}$ below.
 \begin{algorithm}[H]
 	\caption{Exploring potential qubits in error for toric codes with periodic boundaries} \label{alg:tcc-nobound-potqubits}
 	\begin{algorithmic}[1]
 		\REQUIRE A 3D complex $\Gamma$, collection of edges $S_E$ carrying nonzero syndrome and an artificial boundary $\mathcal{X}$ as given in Eq.~\eqref{eq:chi}.
 		\ENSURE Collection of faces $\mathcal{E} \subseteq C_2(\Gamma) \setminus \mathcal{X}$ such that there exists $ \mathcal{E}' \subseteq \mathcal{E}\cup \mathcal{X}$ and $ \partial (\mathcal{E}') = S_E$
 		\STATE Initialize the boundary as $B=S_E$, $\mathcal{E}= \emptyset$ and mark all faces as unexplored.
 		\WHILE{there exist unexplored faces in $C_2(\Gamma) \setminus \mathcal{X}$}
 		\STATE	$B'=\emptyset$ \COMMENT{Boundary for next stage}
 		\FOR{all unexplored faces $f \in C_2(\Gamma) \setminus \mathcal{X}$ incident on the boundary $B$}
 		\STATE Mark the face $f$ as explored.
 		\IF{ the set $\mathcal{X} \cup \mathcal{E} \cup \{f\} $ is not a cut set in $\Gamma$} 
 		\STATE Update $\mathcal{E}=\mathcal{E} \cup \{ f\}$ \COMMENT{ $f$ is a potential qubit in error}
 		\STATE Update $B'=B'\cup \partial(f)$ 
 		\ENDIF
 		\ENDFOR
 		\STATE $B=B'\setminus B$ 
 		\ENDWHILE
 		\STATE Return $\mathcal{E}$ and exit
 	\end{algorithmic}
 \end{algorithm}
 
 \begin{lemma}\label{lm:prop-induced-erasures-nobound}
 	Let $\mathcal{E}$ be the set of faces returned by   Algorithm~\ref{alg:tcc-nobound-potqubits}. Then the following properties hold
 	\begin{compactenum}[(a)]
 		\item If $A$ is a nonempty collection of faces such that $X_{A}$ is a stabilizer, then $A \nsubseteq \mathcal{E} \cup \mathcal{X}$. If $X_A$ is a logical operator and  $A \nsubseteq \mathcal{X}$, then $A \nsubseteq \mathcal{E} \cup \mathcal{X}$.
 		\item If any face $f \notin \mathcal{E} \cup \mathcal{X}$, 
 		then $ \mathcal{X} \cup \mathcal{E}\cup \{ f \}$ supports a stabilizer or a logical operator $X_{B}$
 		where $B\subseteq \mathcal{E}\cup \{ f \}$ and 
 		$f\in B$.
 	\end{compactenum}
 \end{lemma}
 \begin{proof}
 	The proof is similar to that of Lemma \ref{lm:prop-induced-erasures}. 
 \end{proof}

 The previous algorithm gives a potential set of qubits which can explain the measured syndrome. 
 Our next task is to find the qubits in $\mathcal{E}$ which are in error. 
 This we address in the following section.

 \subsection{Projecting the error to the boundary $\mathcal{X}$}
 
 We now show the existence of a solution in $\mathcal{X} \cup \mathcal{E}$ so that after peeling along $\mathcal{E}$, the  error pattern has been projected to $\mathcal{X}$.
 Later we can find the remaining part of error which is  projected to $\mathcal{X}$.
 \begin{theorem}
 	\label{lm:nobound-peel-exist}
 	Let $S_E$ be the support of syndrome  and $\mathcal{E}$ the set of faces returned by  Algorithm~\ref{alg:tcc-nobound-potqubits}. 
 	\begin{compactenum}[(a)]
 		\item {Existence:} There exists a set of faces $\hat{\mathcal{E}} \subset \mathcal{E} \cup \mathcal{X}$ such that $\partial(\hat{\mathcal{E}}) = S_E $.
 		\item {Uniqueness:} If there are two sets of faces $\hat{\mathcal{E}}_1, \hat{\mathcal{E}}_2 \subset \mathcal{E} \cup \mathcal{X}$ such that $\partial(\hat{\mathcal{E}}_1)= \partial(\hat{\mathcal{E}}_2) = S_E $, then $\hat{\mathcal{E}}_1 \triangle \hat{\mathcal{E}_2} \subset \mathcal{X}$. 
 		In other words $\hat{\mathcal{E}}_1 \setminus \mathcal{X} = \hat{\mathcal{E}}_2 \setminus \mathcal{X}$. 
 	\end{compactenum}
 \end{theorem}
 \begin{proof}
 	\begin{compactenum}[(a)]
 		\item 
 		First part of the proof is similar to that of  Lemma~\ref{lm:peel-bound}.
 		We have $S_E$ as the support of nonzero syndrome due to an $X$-type error with support $F\subseteq C_2(\Gamma)$. 
 		If $F \subseteq \mathcal{E} \cup \mathcal{X}$, then the existence  is shown. 
 		If $F \nsubseteq \mathcal{E} \cup \mathcal{X}$, then there exists a face $f \in F \setminus \mathcal{E} \cup \mathcal{X}$.
 		By Lemma~\ref{lm:prop-induced-erasures-nobound}(b),   $f$ is not added to $\mathcal{E}$ because there exists a stabilizer or logical operator with support $B \subseteq \{f\} \cup \mathcal{E} \cup \mathcal{X}$. 
 		Let $F'=F\Delta B$. 
 		Then the operator $X_{F'}$ has the same syndrome as $X_F$ and it  does not have any support on $f$.
 		Further, $F'\setminus \mathcal{E}\cup \mathcal{X}$ is a proper subset of $F\setminus \mathcal{E}\cup \mathcal{X}$.
 		We can repeat this process with $F'$ until all the faces $F \setminus \mathcal{E} \cup \mathcal{X}$ have been removed.
 		At this point we have a solution that is entirely in the support of $\mathcal{E}\cup \mathcal{X}$.
 		\item 
 		If $\hat{\mathcal{E}}_1$, and $\hat{\mathcal{E}}_1$ have the same boundary, then 
 		$X_{\hat{\mathcal{E}}_1}$ and $X_{\hat{\mathcal{E}}_2}$ have the same syndrome.
 		Then $X_{\hat{\mathcal{E}}_1}X_{\hat{\mathcal{E}}_2}$ must be a stabilizer or a logical operator with support $\hat{\mathcal{E}}_1\triangle \hat{\mathcal{E}}_2$. If we have  $\hat{\mathcal{E}}_1, \hat{\mathcal{E}}_2 \subseteq  \mathcal{E} \cup \mathcal{X}$, then $\hat{\mathcal{E}}_1 \triangle \hat{\mathcal{E}}_2  \subseteq \mathcal{E} \cup \mathcal{X}$ which is the support of a stabilizer or a logical operator. 
 		If $\hat{\mathcal{E}}_1 \triangle \hat{\mathcal{E}}_2 $
 		is the support of a stabilizer, then by 
 		Lemma~\ref{lm:prop-induced-erasures-nobound}, $\hat{\mathcal{E}}_1 \triangle \hat{\mathcal{E}}_2 \nsubseteq
 		\mathcal{E}\cup \mathcal{X}$ contradicting our previous conclusion.
 		If $\hat{\mathcal{E}}_1 \triangle \hat{\mathcal{E}}_2 $
 		is the support of a logical operator, then by Lemma~\ref{lm:prop-induced-erasures-nobound},  $\hat{\mathcal{E}}_1 \triangle \hat{\mathcal{E}}_2  \subseteq \mathcal{X}$. 
 	\end{compactenum}
 \end{proof}
 Now, we can find an estimate of the error in $C_2(\Gamma)\setminus \mathcal{X}$ using the following algorithm.
 It also makes use of peeling algorithm given earlier.
 There are some modifications, we give the details below. 
 \begin{algorithm}[ht]
 	\caption{Peeling for toric codes without boundaries} \label{alg:peeling-nobound}
 	\begin{algorithmic}[1]
 		\REQUIRE  $\Gamma$: A 3D lattice, 
 		$\mathcal{X}$: an artificial boundary, 
 		$S_E$: Edges carrying 
 		nonzero syndrome, 
 		$\mathcal{E}$: potential faces in error.
 		\ENSURE A collection of faces $\hat{\mathcal{E}} \subset \mathcal{E}$ and an updated syndrome $S_\mathcal{X}$ such that there exists a collection of faces $\mathcal{E}' \subset \mathcal{X}$ and $\partial(\mathcal{E}') = S_\mathcal{X}$ and 
 		$S_E'$.
 		\STATE Initialize the error estimate support $\hat{\mathcal{E}} =\emptyset$.
 		\WHILE{there exists an edge $e$ which is incident on only one face $f \in \mathcal{E}\cup \mathcal{X}$} 
 		\IF{ $e$ carries nonzero syndrome } \COMMENT{ $e \in S_E$}
 		\STATE $\hat{\mathcal{E}} = \hat{\mathcal{E}}\cup \{ f \}$ 
 		\STATE $S_E =S_E \Delta \partial(f)$ \COMMENT{ Update syndrome} 
 		\ENDIF
 		\STATE
 		$ \mathcal{E} = \mathcal{E} \setminus \{ f\}$ 
 		\COMMENT Remove $f$ from $\mathcal{E}$.
 		
 		\ENDWHILE
 		\IF{ nonzero syndromes are cleared in the interior  }\COMMENT{ ($S_E \subseteq \cup_{f \in \mathcal{X}} \partial(f)$}
 		\STATE Return $\hat{\mathcal{E}}$, $S_\mathcal{X}=S_E$ and exit.
 		\ELSE
 		\STATE Return decoder failure
 		\ENDIF
 	\end{algorithmic}
 \end{algorithm}
 After the completion of Algorithm~\ref{alg:peeling-nobound}, 
 either there is a decoder failure or the error has been projected to $\mathcal{X}$.
 For a valid input, Algorithm~\ref{alg:peeling-nobound}
 will fail  if the syndrome has not been projected to $\mathcal{E}$.
 The failure can happen if the error takes the shape of Klein bottle-like structure. 
 
 \subsection{Estimating residual error in $\mathcal{X}$}
 
 At this stage the residual error is restricted to $\mathcal{X}$.
 We can estimate this error, using the same ideas as outlined before in Section~\ref{sec:intuition}.
 By design $\mathcal{X}$ does not contain any stabilizer. Therefore, multiple solutions in a set of potential faces $\mathcal{E}$ are due to the presence of logical operators alone.
 We avoid multiple solutions by exhaustively checking for the presence of all logical operators within $\mathcal{E}\subseteq \mathcal{X}$.
 For the lattices that we consider, the number of logical operators is $O(1)$.
 So  checking for the presence of logical operators can be performed efficiently. 
 The complete procedure is given in Algorithm~\ref{alg:cleaningchi}.
 \begin{algorithm}[H]
 	\caption{Estimating residual error in $\mathcal{X}$.} \label{alg:cleaningchi}
 	\begin{algorithmic}[1]
 		\REQUIRE A 3D lattice $\Gamma$, Artificial boundary $\mathcal{X}$, and the collection of edges carrying non zero syndromes $S_E$.
 		\ENSURE Collection of faces such that $\hat{\mathcal{E}} \subset \mathcal{X}$ with $S_E = \partial(\hat{\mathcal{E}})$.
 		\STATE Initialize $\mathcal{E} =  \emptyset$, $\mathcal{E}_{\mathcal{X}} = \mathcal{X}$ and $B = S_E$.
 		\WHILE{ $\mathcal{E}_{\mathcal{X}} \neq \emptyset$ }
 		\STATE $B' = \emptyset $
 		\FOR{ all faces $f \in \mathcal{E}_{\mathcal{X}}$ and $\partial(f) \cap B \neq \emptyset$ }
 		\IF{ $\mathcal{E} \cup \{ f \}$ do not support a logical operator }
 		\STATE Update $\mathcal{E} = \mathcal{E} \cup \{f\}$.
 		\STATE $B' = B' \cup \partial(f)$
 		\ENDIF
 		\STATE Update $\mathcal{E}_{\mathcal{X}} = \mathcal{E}_{\mathcal{X}} \setminus \{f\}$.
 		\ENDFOR
 		\STATE $B = B' \setminus B$
 		\ENDWHILE
 		\STATE Initialize $\hat{\mathcal{E}}=\emptyset$ , $B= S_E$.
 		\WHILE{there exists an edge $e$ such that it is incident on only one face $f \in \mathcal{E}$ }
 		\IF{ $e \in B$ }
 		\STATE $\hat{\mathcal{E}} =\hat{\mathcal{E}} \cup \{f\} $ and $B = B \Delta \partial(f)$
 		\ENDIF
 		\STATE $\mathcal{E} = \mathcal{E} \setminus \{f\}$  
 		\ENDWHILE
 		\STATE Return $\hat{\mathcal{E}}$.
 	\end{algorithmic}
 \end{algorithm}

 For toric codes on  cubic lattices and some similar lattices, we have an alternative algorithm. 
 Those lattices should posses the following properties: 
 \begin{compactenum}[(i)]
 	\item We can pick $\mathcal{X} = \bigcup\limits_{i \in \mathcal{L}} \supp(\bar{X_i}^r)$ such that for all independent logical operators $\bar{X}_i^r$ and $\bar{X}_j^r$,  $\supp(\bar{X_i}^r) \cap \supp(\bar{X_j}^r) = \emptyset$. 
 	
 	\item Let $E_{>2}$ be the collection of edges in $\mathcal{X}$ that are incident on more than two faces. 
 	The edges in $E_{>2}$  
 	do not contain a homologically trivial cycle. 
 \end{compactenum}
 
 For example, see the choice of $\mathcal{X}$ given in Fig. \ref{fig:chione}. The three independent logical operators within $\mathcal{X}$ are colored differently and it can be seen that they do not intersect. 
 If we remove the edges that are incident on more than two faces, then $\mathcal{X}$ will be split into three disconnected surfaces $\mathcal{X}_i$
 where 
 \begin{eqnarray}
 	\mathcal{X}_i=\supp(\bar{X}_i^r),\label{eq:chi-i}
 \end{eqnarray}
 for $i=1,2,3$. 
 We can decode within these $\mathcal{X}_i$ independently ignoring the edges incident on more than two faces. 
 This is because $\mathcal{X}_i\cap\mathcal{X}_j =\emptyset$.
 The second property implies that 
 no collection of faces exists in $\mathcal{X}$ such that its boundary is within $E_{>2}$.  
 These edges turn out to be dependent checks and can be ignored. 
 In any surface $\mathcal{X}_i$, there are exactly two solutions.
 This can be easily computed in time $|\mathcal{X}_i|$ for each partial boundary $\mathcal{X}_i$ as in Fig.~\ref{fig:intuition}. 
 We give the complete details for this special case in Algorithm~\ref{alg:cleaningchicubic}. 
 
 \begin{algorithm}[H]
 	\caption{Estimating residual error in $\mathcal{X}$ for cubic lattice.} \label{alg:cleaningchicubic}
 	\begin{algorithmic}[1]
 		\REQUIRE A 3D lattice $\Gamma$, Artificial boundary $\mathcal{X}$, $\supp{\bar{X}_i^r},i=1,2,3$, and the collection of edges carrying non zero syndromes $S_E$.
 		\ENSURE Collection of faces such that $\hat{\mathcal{E}} \subset \mathcal{X}$ with $ \partial(\hat{\mathcal{E}}) = S_E$.
 		\STATE $E_{>2} = \{e \in C_1(\Gamma): |\iota(e) \cap \mathcal{X}| >2 \}$ 
 		\FOR{$i=1,2,3$}
 		\STATE $S_E' = \left(S_E \cap \bigcup\limits_{f \in \supp(\bar{X}_i^r)} \partial(f)\right) \setminus E_{>2} $
 		\STATE Initialize $\mathcal{E} = \supp(\bar{X}_i^r) \setminus \{f\} $  \COMMENT{ Unexplored faces}
 		\STATE Initialize $\hat{\mathcal{E}}_i = \{f\}$ for some $f \in \supp(\bar{X}_i^r)$ \COMMENT{ Estimate for $i$th plane}
 		\STATE Initialize $B' =  S_E' \Delta (\partial(f) \setminus E_{>2}) $. \COMMENT{ Updated syndrome}
 		\WHILE{there exist $f' \in \mathcal{E}$ such that $\partial(f')  \cap \partial(\mathcal{E}) \setminus E_{>2} \neq \emptyset$}
 		\IF{ $\partial(f') \cap \partial(\mathcal{E}) \setminus E_{>2} \subseteq B'$}
 		\STATE Update $\hat{\mathcal{E}}_i =\hat{\mathcal{E}}_i \cup \{f'\}$
 		\STATE Update $B' = B' \Delta (\partial(f') \setminus E)_{>2}$
 		\ENDIF 
 		\STATE Update $\mathcal{E} = \mathcal{E} \setminus \{f'\}$.
 		\ENDWHILE 
 		\IF{$|\hat{\mathcal{E}}_1| > \frac{1}{2} |\supp{\bar{X}_i^r}|$}
 		\STATE $\hat{\mathcal{E}}_1 = \supp{\bar{X}_i^r} \setminus \hat{\mathcal{E}}_1$
 		\ENDIF
 		\ENDFOR
 		\STATE Return $\hat{\mathcal{X}} = \cup_i \hat{\mathcal{X}}_i$.
 	\end{algorithmic}
 \end{algorithm}

 \subsection{Complete decoding algorithm for periodic 3D toric codes}

 To summarize, to decode 3D toric codes without boundary, we introduce an artificial boundary $\mathcal{X}$. 
 We first correct the errors in the rest of the lattice, i.e.,  $C_2(\Gamma)\setminus \mathcal{X}$ using Algorithm~\ref{alg:tcc-nobound-potqubits} and \ref{alg:peeling-nobound}.
 This leaves the errors on $\mathcal{X}$ uncorrected. 
 In effect, we  project the error onto $\mathcal{X}$.
 We then proceed to correct the errors remaining in $\mathcal{X}$ 
 using Algorithm~\ref{alg:cleaningchicubic} or \ref{alg:cleaningchi}.
 Whenever, there is a decoding failure we can run the decoder another time with a different $\mathcal{X}$. 
 Typically a decoding failure due to the presence of Klein bottle-like structure is not observed in the second run.

 The algorithm converges and finds one $E$ such that $\partial(E)=S$. The choice is made on a greedy approach. 
 
 \begin{algorithm}[H]
 	\caption{Decoding 3D toric codes with periodic boundaries} \label{alg:ovall}
 	\begin{algorithmic}[1]
 		\REQUIRE 3D lattice $\Gamma$, collection of edges $S_E$ with nonzero syndrome.
 		\ENSURE Collection of faces $\hat{\mathcal{E}}$ such that $ \partial(\hat{\mathcal{E}}) = S_E$.
 		\STATE Inititalize $\hat{\mathcal{E}}' = \emptyset $
 		\STATE \label{steps:int} Find potential qubits in error using Algorithm~\ref{alg:tcc-nobound-potqubits} with nonzero syndromes $S_E$ and artificial boundary $\mathcal{X}$ and obtain $\mathcal{E}$.
 		\STATE  \label{steps:peel} Do peeling using Algorithm~\ref{alg:peeling-nobound} over $\mathcal{E}$ and $S_E$ and obtain $\hat{\mathcal{E}}$.
 		\STATE\label{steps:upd} Update $\hat{\mathcal{E}}' = \hat{\mathcal{E}}' \cup \hat{\mathcal{E}}$ and $S_E = S_E \Delta \partial(\hat{\mathcal{E}})$.
 		\IF{$S_E \nsubseteq \bigcup\limits_{f\in \mathcal{X}} \partial(f)$}
 		\STATE Repeat lines \ref{steps:int}- \ref{steps:upd} using different $\mathcal{X}$. 
 		\ENDIF
 		\STATE Clear residual syndrome on $\mathcal{X}$ using Alg. \ref{alg:cleaningchi} and obtain $\hat{\mathcal{E}}$. 
 		\STATE Update $\hat{\mathcal{E}}' = \hat{\mathcal{E}}' \cup \hat{\mathcal{E}}$ and $S_E = S_E \Delta \partial(\hat{\mathcal{E}})$.
 		\IF{ $S_E = \phi$} 
 		\STATE Return $\hat{\mathcal{E}} =\hat{\mathcal{E}}'$ and exit.
 		\ELSE
 		\STATE Report decoder failure and exit.
 		\ENDIF
 	\end{algorithmic}
 \end{algorithm}

 The computational complexity of the Algorithm~\ref{alg:ovall} is $O(n^2)$ where $n$ is number of qubits. 
 This is because at each stage of finding potential qubits in error, we go through all the $n$ faces and check for stabilizer or logical operator for a worst case of $n$ faces. This makes the complexity of algorithm as $O(n^2)$. 
 The running time of all the other steps in the decoding is less than $n^2$. Hence the overall complexity of our algorithm is $O(n^2)$.

 We can reduce the complexity of decoding by the following modification.
 In Algorithm~\ref{alg:tcc-nobound-potqubits}, lines 6--9 can
 be modified as follows. 
 Instead of checking if the set $\mathcal{X}\cup \mathcal{E}\cup \{ f\}$ is a cut set, we go ahead and add the face to $\mathcal{E}$ but we keep track of the sequence in which it is added to $\mathcal{E}$.
 When all faces have been explored, starting 
 from any volume $\nu$ we  remove faces  which are added last so that the there exists a unique face path from $\nu$ to any other volumes through these removed faces. Thus we can get rid of cutset and also freezing these removed faces which are added last to $\mathcal{E}$.

 This ensures that the updated set $\mathcal{E}$ will not be a cut set.
 We conjecture that the resulting complexity will be superlinear but subquadratic. 
 
 \section{Simulation details}\label{sec:sims}
 
 We simulated the decoding algorithm for the toric code on a cubic lattice of size $L\times L\times L $ with periodic boundaries. 
  The simulation results are plotted in Fig.~\ref{fig:simres}. 
 Each sample in the figure is obtained by running the algorithm repeatedly for $10^5$ iterations or 300 logical errors. 
 The choice of $\mathcal{X}$ was made as shown in Fig.~\ref{fig:chione}. 
 Only a negligible number of Klein bottlelike structures were observed in the first iteration. These failures vanished after the second iteration. 
 We obtained  a threshold of $~12.2\%$ for bit flip errors.
 This is comparable to the results obtained in 
 \cite{kulkarni18,kubica18}.
 Other decoders \cite{duivenvoorden18,breuckmann18} that are optimized for the cubic lattice perform better with threshold of $17.2\%$ and $17.5\%$.

 We have also used this decoder to decode stacked color codes \cite{connor16}. 
 These codes can be projected to toric codes.
 The toric codes obtained there contain boundaries and the faces consist only of triangles. 
 We did not  notice decoding failures due peeling
 while simulating those those codes.
 For the bit flip channel we obtained a threshold 
 of $13.2\%$, see \cite{aloshious18}.
 
 \begin{figure}[H]
 	\centering
 	\includegraphics[scale=0.8]{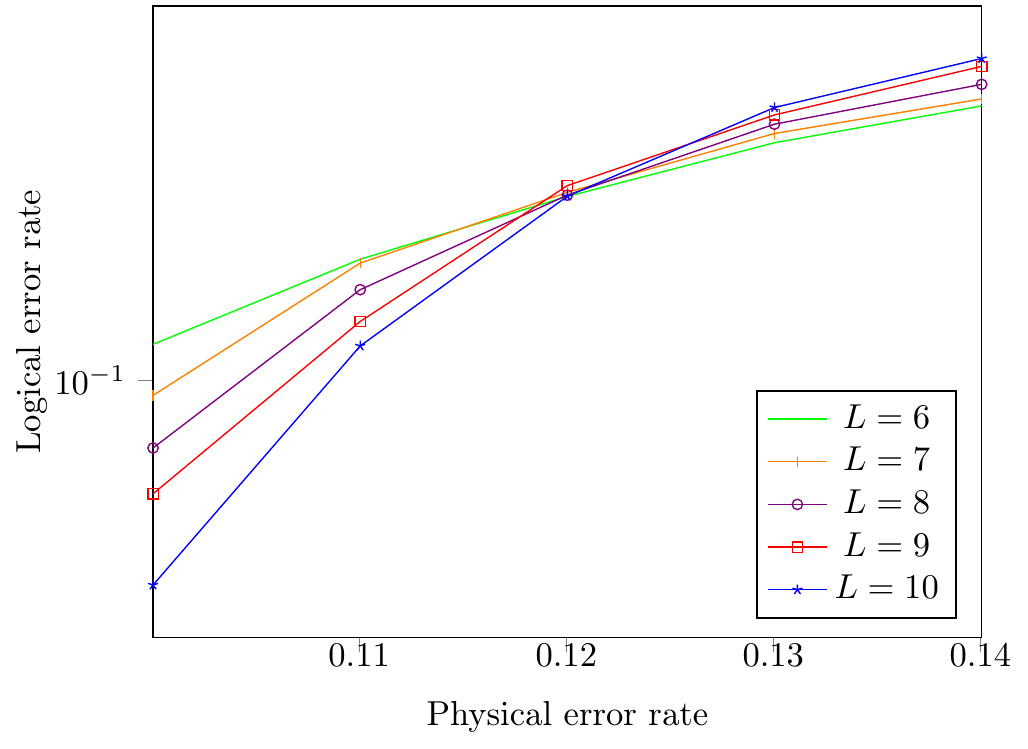}
 	\caption{Simulation results on performance of bit flip error in 3D toric code with periodic boundaries. Length of the code for each block length $L$ is $3 \times L^3$.}
 	\label{fig:simres}
 \end{figure}
 
 \section{Conclusion}
 
 We proposed a 3D toric code decoder for the bit flip channel.
 This implicitly transforms the decoding problem to a decoding problem over an erasure channel.
The proposed decoder is applicable to a large class of lattices
 including  certain 3D lattices for which the method proposed in 
 \cite{kubica18} may not be applicable.
\nix{ Our algorithm can be thought be a type of greedy algorithm.}
 A fruitful direction for further research is to incorporate syndrome measurement errors
 during decoding. 
 Another avenue for further research is to determine  bounds for the performance of the algorithm.


\begin{thebibliography}{10}
 	
 	\bibitem{alicki10}
 	R.~Alicki, M.~Horodecki, P.~Horodecki, and R.~Horodecki.
 	\newblock On thermal stability of topological qubit in kitaev's 4d model.
 	\newblock {\em Open Systems \& Information Dynamics}, 17(01):1--20, 2010.
 	
 	\bibitem{aloshious18}
 	Arun~B. Aloshious and Pradeep~Kiran Sarvepalli.
 	\newblock Projecting three-dimensional color codes onto three-dimensional toric
 	codes.
 	\newblock {\em Phys. Rev. A}, 98:012302, Jul 2018.
 	
 	\bibitem{breuckmann18}
 	Nikolas~P. Breuckmann and Xiaotong Ni.
 	\newblock Scalable {N}eural {N}etwork {D}ecoders for {H}igher {D}imensional
 	{Q}uantum {C}odes.
 	\newblock {\em {Quantum}}, 2:68, May 2018.
 	
 	\bibitem{calderbank98}
 	A.~R. {Calderbank}, E.~M. {Rains}, P.~W. {Shor}, and N.~J.~A. {Sloane}.
 	\newblock Quantum error correction via codes over gf(4).
 	\newblock In {\em Proceedings of IEEE International Symposium on Information
 		Theory}, pages 292--, June 1997.
 	
 	\bibitem{castelnovo08}
 	Claudio Castelnovo and Claudio Chamon.
 	\newblock Topological order in a three-dimensional toric code at finite
 	temperature.
 	\newblock {\em Phys. Rev. B}, 78:155120, Oct 2008.
 	
 	\bibitem{delfosse17}
 	Nicolas Delfosse and Naomi~H. Nickerson.
 	\newblock {Almost-linear time decoding algorithm for topological codes}.
 	\newblock {\em Arxiv: 1709.06218 }
 	\newblock 2017.
 	
 	\bibitem{delfosse17a}
 	Nicolas Delfosse and Gilles Z{\'{e}}mor.
 	\newblock {Linear-Time Maximum Likelihood Decoding of Surface Codes over the
 		Quantum Erasure Channel}.
 	\newblock {\em Arxiv: 1703.01517 }
 	\newblock 2017.
 	
 	\bibitem{dennis02}
 	Eric Dennis, Alexei Kitaev, Andrew Landahl, and John Preskill.
 	\newblock Topological quantum memory.
 	\newblock {\em Journal of Mathematical Physics}, 43(9):4452--4505, 2002.
 	
 	\bibitem{duivenvoorden18}
 	K.~{Duivenvoorden}, N.~P. {Breuckmann}, and B.~M. {Terhal}.
 	\newblock Renormalization group decoder for a four-dimensional toric code.
 	\newblock {\em IEEE Transactions on Information Theory}, 65(4):2545--2562,
 	April 2019.
 	
 	\bibitem{gottesman97}
 	Daniel Gottesman.
 	\newblock {\em {Stablizer Codes and Quantum Error correction}}.
 	\newblock Phd thesis, California Institute of Technology, 1997.
 	
 	\bibitem{hamma05}
 	Alioscia Hamma, Paolo Zanardi, and Xiao-Gang Wen.
 	\newblock String and membrane condensation on three-dimensional lattices.
 	\newblock {\em Phys. Rev. B}, 72:035307, Jul 2005.
 	
 	\bibitem{kubica19}
 	Aleksander Kubica and Nicolas Delfosse.
 	\newblock {Efficient color code decoders in d $\geq$ 2 dimensions from toric
 		code decoders}.
 	\newblock {\em Arxiv: 1905.07393 }
 	\newblock 2019.
 	
 	\bibitem{kubica18}
 	Aleksander Kubica and John Preskill.
 	\newblock {Cellular-automaton decoders with provable thresholds for topological
 		codes}.
 	\newblock {\em Arxiv: 1809.10145 }
 	\newblock 2018.
 	
 	\bibitem{kubica15}
 	Aleksander Kubica, Beni Yoshida, and Fernando Pastawski.
 	\newblock Unfolding the color code.
 	\newblock {\em New Journal of Physics}, 17(8):083026, Aug 2015.
 	
 	\bibitem{kulkarni18}
 	Abhishek Kulkarni and Pradeep~Kiran Sarvepalli.
 	\newblock {Decoding the three-dimensional toric codes and welded codes on cubic lattices}.
 	\newblock {\em Phys. Rev. A\/}, {\bf 100}, 012311, 2019.
 	
 	\bibitem{vasmer18}
 	Michael Vasmer and Dan~E. Browne.
 	\newblock {Universal Quantum Computing with 3D Surface Codes}.
 	\newblock {\em Arxiv 1801.04255 }
 	\newblock 2018.
 	
 	\bibitem{bravyi11}
 	Sergey Bravy, Bernhard Leemhuis and Barbara M.Terhal .
 	\newblock {Topological order in an exactly solvable 3D spin model}.
 	\newblock {\em Annals of Physics }
 	\newblock Volume 326, Issue 4, April 2011, Pages 839-866.
 	
 	\bibitem{kubica15a}
 	{Kubica, A.} and {M.~E. Beverland} (2015).
 	\newblock Universal transversal gates with color codes: A simplified approach.
 	\newblock {\em Phys. Rev. A\/}, {\bf 91}, 032330.
 
 \bibitem{bombin15}
 {Bomb\'{\i}n, H.} (2015).
 \newblock Single-shot fault-tolerant quantum error correction.
 \newblock {\em Phys. Rev. X\/}, {\bf 5}, 031043.
 
 \bibitem{brown16}
 {Brown, B.~J.}, {N.~H. Nickerson}, and {D.~E. Browne} (2016).
 \newblock Fault-tolerant error correction with the gauge color code.
 \newblock {\em Nature Communications\/}, {\bf 7}, 12302 EP --.
 
 \bibitem{connor16}
 {Jochym-O'Connor, T.} and { S.~D. Bartlett} (2016).
 \newblock Stacked codes: Universal fault-tolerant quantum computation in a
 two-dimensional layout.
 \newblock {\em Phys. Rev. A\/}, {\bf 93}, 022323.

\bibitem{sullivan90}
John Mathew Sullivan.
\newblock A crystalline approximation theorem for hypersurfaces.
\newblock {\em Princeton University}, Ph.D.Thesis, 1990.
 \end{thebibliography}

%
\end{document}